\documentclass[a4paper,11pt,numbers=noenddot]{article}
\usepackage[a4paper,left=2.5cm,right=2.5cm,top=2.5cm,bottom=2.5cm]{geometry}
\usepackage{graphicx} 
\usepackage[utf8]{inputenc}
\usepackage{comment}
\usepackage{amsmath}
\usepackage{amsthm}
\usepackage{amssymb}
\usepackage{fdsymbol}
\usepackage{dsfont}%
\usepackage{physics}
\usepackage{faktor}
\usepackage[colorlinks=true, allcolors=blue]{hyperref}
\usepackage{bbm}
\usepackage{float}
\usepackage{subfig}
\usepackage{caption}
\usepackage{subcaption}
\usepackage{bbm}
\usepackage{mathrsfs}
\usepackage[dvipsnames]{xcolor}

\newcommand{\R}{\mathbb{R}}
\newcommand{\C}{\mathbb{C}}

\newcommand{\T}{\mathbb{T}}
\newcommand{\Z}{\mathbb{Z}}
\newcommand{\Sp}{\mathbb{S}}
\newcommand{\Hp}{\mathcal{H}}

\newcommand{\N}{\mathbb{N}}

\newcommand{\Pp}{\mathbb{P}}
\newcommand{\E}{\mathbb{E}}

\newcommand{\dist}{ \operatorname{dist}} 

\newtheorem{theorem}{Theorem}
\newtheorem{definition}{Definition}
\newtheorem{lemma}{Lemma}
\newtheorem{proposition}{Proposition}

\theoremstyle{definition}

\newtheorem{remark}{Remark}
\numberwithin{equation}{section}
\usepackage{mathtools}
\mathtoolsset{showonlyrefs}

\usepackage[square,sort,comma,numbers]{natbib}

\usepackage{tikz}
\usetikzlibrary{arrows.meta,positioning}
\usetikzlibrary{decorations.markings}
\tikzset{
  midarrow/.style={
    postaction={decorate},
    decoration={
      markings,
      mark=at position 0.5 with {\arrow{Stealth[length=4pt]}}
    }
  },
  edge/.style={
    thick,
    midarrow
  },
  purpleedge/.style={
    line width=1.6pt,
    color=violet,
    postaction={decorate},
    decoration={
      markings,
      mark=at position 0.5 with {\arrow{Stealth[length=7pt, width=6pt]}}
    }
  },
  blueedge/.style={
    line width=1.6pt,
    color=blue,
    postaction={decorate},
    decoration={
      markings,
      mark=at position 0.5 with {\arrow{Stealth[length=7pt, width=6pt]}}
    }
  }
}
\tikzset{
  node/.style={
    circle,
    draw,
    thick,
    minimum size=18pt,
    inner sep=1pt
  }
}

\usepackage[none]{hyphenat}
\usepackage{authblk}
\begin{document}
\title{Dynamical Localization for\\ General Scattering Quantum Walks}

\author[1]{Alain Joye\footnote{alain.joye@univ-grenoble-alpes.fr}}
\affil[1]{Univ. Grenoble Alpes,
CNRS, Institut Fourier,
F-38000 Grenoble France
}
\author[1]{Andreas Schaefer\footnote{andreas.schaefer@univ-grenoble-alpes.fr}}

\author[2,3]{Simone Warzel\footnote{simone.warzel@tum.de}}
\affil[2]{Departments of Mathematics and Physics, TU Munich, Germany
}
\affil[3]{Munich Center for Quantum Science and Technology, Munich, Germany}

\date{\today}
\maketitle

\begin{abstract}
    We consider quantum walks defined on arbitrary infinite graphs, parameterized by a family of scattering matrices attached to the vertices. Multiplying each scattering matrix by an i.i.d.\ random phase, we obtain a random scattering quantum walk. 
    We prove dynamical localization for random scattering walks in a large-disorder regime. The result is based on a relation between fractional moment estimates and eigenfunction correlators of independent interest, which we establish for general random unitary operators.
\end{abstract}

\setcounter{tocdepth}{2}
\tableofcontents

\section{Introduction}

In the last few decades, quantum walks on graphs, which are discrete-time linear dynamical systems, have become increasingly popular as they lie at the crossroads of different scientific fields \cite{AAKV:2001, Kempe:2003, VenegasAndraca:2012, ABJ:2015, GZ:2023}. They provide simplified, yet sensible, models for the quantum dynamics of different physical systems, {\it e.g.}\ in condensed matter physics, quantum optics or quantum statistical mechanics; see  \cite{Chalker:1988,  BHJ:2003, KOK:2005, FH:2005, TMT:2020, WM:2013, AJR:2021, BJS:2025} and references therein. Quantum walks also appear within the context of quantum information and algorithms, in particular as a versatile framework for designing and implementing search algorithms that are more efficient than classical ones; see \cite{Grover:1996, Shenvi:2003, Childs:2009, Portugal:2013, Watrous:2018, Qiang:2024}. From a probabilist's perspective, quantum walks can be viewed as non-commutative analogs of classical Markov chains, a viewpoint that has sparked extensive mathematical analysis of (open) quantum walks by comparison with the classical case, \cite{Szegedy:2004,  Gudder:2008, APSS:2012, Konno:2008, GVWW:2013, Tiedra:2021, joye:2024}, to give a few examples. Let us mention that CMV matrices \cite{CMV:2003}, the unitary counterparts of Jacobi matrices in the self-adjoint setting, which play a pivotal role in the theory of orthogonal polynomials and the spectral analysis of unitary operators \cite{Simon:2005a, Simon:2005b}, are special cases of quantum walks. Quantum walks may also exhibit topological properties with consequences on their spectral and dynamical behavior, as shown for example in \cite{Cedzich:2017, DFT:2017, SS-B:2017, Richard:2017, ABJ:2019, ABJ:2020, CGG+:2018, Delplace2020}. 

The line of research to which the present contribution belongs is concerned with the effect of disorder on the dynamical properties of quantum walks. In applications, randomness may arise from various sources: in condensed matter physics, the environment in which the quantum walker evolves is often disordered, while in quantum information theory, the randomness stems from errors or lack of control in the implementation of the algorithms. In such situations, we speak of a random quantum walk, and the objective is to study the hindrance of the propagation of the quantum walker induced by disorder, {\it i.e.}\ localization. Mathematical analyses of (de-)localization for quantum walks subjected to random alterations have been led in various contexts, see {\it e.g.} \cite{Koshovets:1991, Joye:2005, HJS:2006, Simon:2006, HJS:09, JM:2010,  ASW:2011, HJ:2014, ABJ:2012, Joye:2012, BoumazaMarin:2015, ShapiroTauber:2019,  Schaefer:2025, BoumazaKhouildi:2025}. Works addressing the spectral properties of walks characterized by quasi-periodic variations of parameters include \cite{BHJ:2003, WangDamanik:2019, CedzichWerner:2021, CFO:2023, CFLOZ:2024}. 

In this paper, we study scattering quantum walks on infinite graphs propagating in a disordered environment. Scattering quantum walks, characterized by a family of scattering matrices attached to the vertices of the graph, provide a unifying framework for the study of a broad class of quantum walks on arbitrary graphs, including coined quantum walks and network models; see \cite{joye:2024}. We focus on an environment with a minimal amount of randomness, namely, one random phase per scattering matrix only.    

Our goal is twofold. First, we develop a general method for concluding dynamical localization from fractional moment localization results, in an abstract unitary framework. Second, as an application, we prove dynamical localization for random scattering quantum walks defined on arbitrary graphs, with minimal disorder. Specifically, we establish exponential fractional moment estimates in the strong-disorder regime for random scattering quantum walks on arbitrary graphs with a reduced number of random variables, {\it i.e.}\ one random phase per scattering matrix only. 
The scarcity of randomness prevents us from concluding dynamical localization via a second moment estimate based on rank-one analysis as in \cite{HJS:09}. This difficulty motivated the development of a general approach that allows one to conclude dynamical localization from fractional moment bounds via eigenfunction correlators. Unlike previous methods \cite{Warzel:15, HJS:09}, our approach does not rely on rank-one analysis. 
This method is expected to be applicable in other contexts where fractional moment results are available, but second moment estimates are not, see for example \cite{Klausen:2023}.

\section{Eigenfunction correlators for random unitaries} \label{sec:EC}
\subsection{Definitions and basic properties} \label{sec:defEC}
The eigenfunction correlators (EC for short) of a unitary operator $ U $, which is defined on a separable Hilbert space $ \mathcal{H} $, are simply the total variation measures of the operator's spectral measure $ \mu_{\psi,\varphi} $ associated with normalized vectors $ \psi, \varphi \in  \mathcal{H}  $, i.e., for any open $ I \subset \mathbb{S}^1 $
\begin{equation}\label{def:EC}
    Q(\psi,\varphi;I) \equiv | \mu_{\psi,\varphi} |(I) \coloneqq \sup_{\substack{ F\in C_c(I) \\ \| F \|_\infty \leq 1 }}  \left|\langle \psi | F(U)  \varphi \rangle \right| . 
\end{equation}
This quantity shares most of the properties with its self-adjoint counterpart \cite{AizenmanSchenker:2001,Warzel:15}. Since the proofs of these properties are straightforward extensions of the self-adjoint case, we summarize them only briefly:
\begin{enumerate}
\item 
As is evident from the definition: $  Q(\psi,\varphi;I) \leq 1 $. 
    \item 
By Lusin's theorem, the supremum can also be taken over the set of measurable $ F $ with the property $F = F \ 1_I $ with $ \| F \|_\infty\leq 1$ and $ 1_I $ denoting the indicator function. As a particular case, one thus concludes
\begin{equation} \label{eq: loc from ec decay}
    \sup_{n \in \mathbb{Z} }  \left|\langle \psi | U^n P_I(U)  \varphi \rangle \right| \leq  Q(\psi,\varphi;I) .
\end{equation}
The EC thus captures information about the time-discrete dynamics generated by $ U $, when restricted to the spectral subspace associated with $ I $.
\item By the RAGE theorem~\cite{RichardTiedra:2022}, if 
 $ (\psi_j)_{j \in \mathbb{N} } $ stands for any orthonormal basis, then for any normalized $ \varphi \in \mathcal{H} $ the spectral projection $ P^c_I(U) $ of the unitary on its continuous spectrum within $ I $ is given by 
\begin{align*}
    \langle \varphi | P_I^c(U)  \varphi \rangle =  \lim_{L\to \infty } \lim_{N\to \infty} \frac{1}{N} \sum_{n=0}^{N-1} \sum_{j=L}^\infty |\langle \psi_j | U^n P_I(U) \varphi \rangle |^2 
\end{align*}
It is hence straightforward that if $ \sum_{j\in \mathbb{N}} Q(\psi_j,\varphi;I)^2 < \infty $ for some $ \varphi $, the continuous component of the spectral measure associated with that vector vanishes, $ \langle \varphi | P_I^c(U)  \varphi \rangle =  0  $, cf.~\cite[Thm. 7.2]{Warzel:15}.
\item\label{ec:4}
If a sequence of spectral measures converges weakly, i.e.\ 
$$ \mu^{(n)}_{\psi,\varphi}(F) \to \mu_{\psi,\varphi}(F) \coloneqq \int_{\mathbb{S}^1} F(s) \ d\mu_{\psi,\varphi}(ds)  $$ for all $ F \in C(\mathbb{S}^1) $, then the eigenfunction correlators satisfy $$ | \mu_{\psi,\varphi} |(I) \leq \liminf_{n\to \infty}  | \mu_{\psi,\varphi}^{(n)} |(I)  $$ for all open $ I \subset \mathbb{S}^1 $, cf.~\cite[Prop. 7.6]{Warzel:15}. 
\end{enumerate}

The naming derives from the special case in which the unitary has a discrete spectral decomposition
\begin{equation}\label{decUF}
    U = \sum_\alpha \lambda_\alpha P_\alpha, 
\end{equation}
such that  
$ Q(\psi,\varphi;I)  = \sum_{\lambda_\alpha \in I} | \langle \psi | P_\alpha \varphi \rangle | $. 
If~\eqref{decUF} applies, we may also introduce the family of interpolated eigenfunction correlators 
\begin{equation}\label{def:IEC}
    Q(\psi,\varphi;I,\beta) \coloneqq \sum_{\lambda_\alpha\in I } \langle \psi | P_\alpha \psi \rangle^{1-\beta} | \langle \psi | P_\alpha \varphi \rangle|^{\beta} ,
\end{equation}
with parameter $\beta\in [0,1]$ where $ \psi, \varphi \in \mathcal{H} $ are normalized.
In the self-adjoint case, these quantities were introduced in~\cite{AizenmanSchenker:2001}, see also~\cite{Warzel:15}.
Proofs of the following properties again extend verbatim:
\begin{enumerate}
    \item $ Q(\psi,\varphi;I,0) = Q(\psi,\psi;I) $ and $  Q(\psi,\varphi;I,1) =  Q(\psi,\varphi;I)  $.
    \item In its dependence on $ \beta $, the interpolated EC is a log-convex function, i.e., 
for any $\lambda \in [0,1] $ and $\beta_\lambda \coloneqq (1-\lambda)\beta_0  + \lambda \beta_1 \in [0,1]$:
$$ Q(\psi,\varphi;I,\beta_\lambda)\leq   Q(\psi,\varphi;I,\beta_0)^{1-\lambda} Q(\psi,\varphi;I,\beta_1)^{\lambda} .
$$
Consequently, $ 0 \leq Q(\psi,\varphi;I,\beta) \leq Q(\psi,\varphi;I,\alpha)^{\frac{1-\beta}{1-\alpha}} $  for any $ 0 \leq \alpha \leq \beta \leq 1 $. 
    \item For any $ \beta \in [0,1]$:
    \begin{equation}
    Q(\psi,\varphi;I) \leq \sqrt{Q(\psi,\varphi;I,\beta)Q(\varphi,\psi;I,\beta)} ,
    \end{equation}
    see~\cite[Eq.~(7.44)]{Warzel:15}. 
\end{enumerate}
In particular, an upper bound on the interpolated eigenfunction correlators at one $ \beta \in (0,1] $ implies a bound on any $ \beta \in (0,1] $. \\

Eigenfunction correlators are a useful tool in harvesting dynamical localization properties for self-adjoint random operators~\cite{KuS80,Aizenman94,AizenmanSchenker:2001,Warzel:15}.  
To connect to the localization analysis of random operators, we, however, need additional structure. 
\subsection{Consistent families of random unitaries on subgraphs}\label{sec:consfam}

A rather general framework for random unitaries, which covers random Scattering Quantum Walks (SQW) as a particular case, is described in terms of
graphs $(V,E)$ with vertex set $V$ and countable edge set $E$. We associate random variables to the edges, $ \omega: E \to \mathbb{R}, e \mapsto \omega_e $. They form the canonical probability space $ ( \mathbb{R}^E, \mathcal{B}(\mathbb{R}^E), \mathbb{P} )$.  
The unitaries will be defined on the Hilbert space 
\begin{equation}
    \ell^2(E) \coloneqq \left\{ \psi: E \to \mathbb{C} \, \bigg| \, \sum_{e\in E} |\psi(e)|^2 < \infty\right\} .
\end{equation}
Subsequently, we identify edges $ e \in E $ with the associated canonical orthonormal basis $ | e \rangle $ in this Hilbert space. 
Denoting by $ \mathcal{U}(E)  $ the set of unitary operators on this Hilbert space, weakly measurable maps
\begin{equation} \label{eq:RU} 
U : \mathbb{R}^E \to \mathcal{U}(E) , \, \omega \mapsto U_\omega 
\end{equation}
will be referred to as random unitaries. 

In order to be able to address finite-volume approximations, which are unitary as well, the following notion turns out to be useful. 
\begin{definition} \label{ass1}
We speak of a \emph{consistent family $\mathcal{F} $ of subsets $ F \subset E $} for random unitaries on the Hilbert space $ \ell^2(E) $ of a graph $ (V,E) $ with underlying probability space $( \mathbb{R}^E, \mathcal{B}(\mathbb{R}^E), \mathbb{P} )$ iff 
(i) for any $ F \in  \mathcal{F} $ also its complement $ F^c =  E \setminus F $ is in $ \mathcal{F} $ and 
(ii) there is a random unitary
\begin{equation}\label{eq:finitevolUF}
    U^F: \mathbb{R}^E \to \mathcal{U}( F )  , \ \omega \mapsto U_\omega^F ,
\end{equation}
which is weakly measurable with respect to the sub-sigma algebra $ \mathcal{B}(\mathbb{R}^F) $. 
\end{definition}
\begin{remark} 
The required measurability expresses the fact that $ U^F_\omega $ only depends on the projection of $ \omega $ onto $ \mathbb{R}^F$. 
Given a random unitary~\eqref{eq:RU}, the straightforward choice $ U_\omega^F = P^F U_\omega P^F $ with $ P^F $ the orthogonal projection onto the subspace $ \ell^2(F) $ will generally not be unitary. This is one of the notable differences with the case of random self-adjoint operators for which consistent, self-adjoint Dirichlet restrictions exist for any subset $ F $.  Random SQW naturally lead to consistent families of subsets as will be explained in Section~\ref{sec:RSQW}. 
\end{remark}

For any consistent family $ \mathcal{F} $ of subsets for random unitaries, we have a family of random eigenfunction correlators, $ Q_\omega^F$ with $ F \in \mathcal{F} $ associated with $ U_\omega^F $ and normalized vectors in $ \ell^2(F) $.

In the above set-up, any 
disjoint decomposition of the edge set, $ E = F \uplus F^c $ for $ F \in \mathcal{F} $, comes with the direct sum of unitaries, $ U_\omega^F \oplus U_\omega^{F^c} $, whose 'decoupled' action goes along with the 
orthogonal decomposition of the Hilbert space $  \ell^2(E) = \ell^2(F)  \oplus\ell^2(F^c) $. The boundary operator of $ F$ within $ E $:
\begin{equation}\label{bdryT}
    T_\omega^{E,F} \coloneqq U_\omega^{E} - U_\omega^F \oplus U_\omega^{F^c}  
\end{equation}
compares the decoupled unitaries to the unitary $ U_\omega^{E} \equiv U_\omega $ on the whole graph. 

For any finite subset $ F \in \mathcal{F} $ the random unitary  $ U^{F} $ defines a finite-volume approximation of $ U $. 
In this context, it is natural to inquire about the nature of the convergence of such approximations for exhaustive sequences of finite subsets. This is addressed in the following 
\begin{proposition} \label{propo: EC lower semicont}
Given a sequence of finite subsets $ E_L \in \mathcal{F} $ in a consistent family $\mathcal{F} $ for random unitaries on the Hilbert space over a graph $ (V,E) $ such that the boundary operators converge strongly, i.e., for all $ e \in E $ and all $ \omega \in \mathbb{R}^{E}$
\begin{equation}\label{eq:strTrans}
\lim_{L\to \infty } \| T_\omega^{E,E_L} e \| =0 .
\end{equation} 
Then for all $ e,f \in E$ and all $ \omega \in \mathbb{R}^{E}$:
\begin{enumerate}
    \item\label{item:wconv} the associated spectral measures  converge weakly, i.e., $ (\mu^{E_L}_{\omega})_{e,f}(F) \to  (\mu^{E}_\omega)_{e,f}(F) $ for all $ F\in C(\mathbb{S}^1) $. 
    \item
 the EC is lower semicontinuous, i.e.,  for any open $ I \subset \mathbb{S}^1$
\begin{equation}\label{eq:ECconv}
Q_\omega^{E}(e,f;I) \leq \liminf_{L\to \infty} Q_\omega^{E_L}(e,f;I) .
\end{equation}
\end{enumerate}
\end{proposition}
\begin{proof}
The proof of the first item is based on Lemma~\ref{lem:Us}. The second item follows from item~\ref{ec:4} in Section~\ref{sec:defEC}. 
\end{proof}

\begin{lemma}\label{lem:Us}
    If $\lim_{n\rightarrow \infty}\|(U_n- U)\varphi\|= 0$ for all $\varphi\in \mathcal{H}$, then $\lim_{n\rightarrow}\mu_n(F)=\mu(F)$ for all $F\in C(\mathbb{S}^1)$, where $d\mu_n$, respectively $d\mu$, are the spectral measures of $U_n$, respectively $U$ associated with an arbitrary vector.  
\end{lemma}
\begin{proof}
    By functional calculus and the Stone-Weierstrass theorem, we can approximate any $F\in C(\mathbb{S}^1)$ by a polynomial in the $L^\infty$ norm, so that it is enough to consider monomials. The identity 
    \begin{align}
        U_n^k-U^k=\sum_{j=0}^{k-1} U_n^j(U_n-U)U^{k-j-1}
    \end{align}
    valid for any $ k \in \mathbb{N} $, yields
    \begin{align}
        \|(U_n^k-U^k)\varphi \|\leq \sum_{j=0}^{k-1} \|U_n^j(U_n-U)U^{k-j-1}\varphi \| = \sum_{j=0}^{k-1} \|(U_n-U)(U^{k-j-1}\varphi) \|,
    \end{align}
    where each summand goes to zero by assumption.
\end{proof}

\begin{remark}
As will be explained in Section~\ref{sec:RSQW}, for random SQW
\eqref{eq:strTrans} applies in case $ E_L $ are sequences of balls and the action of boundary operators are restricted to the vicinity of the surface of these balls.
\end{remark}

\subsection{Relation to fractional moments}
The lower semicontinuity~\eqref{eq:ECconv} implies that an upper bound on the EC of a random unitary on a potentially infinite graph $(V,E) $ is given in terms of finite-volume approximations. 
In the following, we may hence restrict our attention on graphs with a finite edge set $ E $. In this case, any unitary has a discrete spectrum~\eqref{decUF}, and we may hence work with the interpolated ECs given by~\eqref{def:IEC}. The following is our first main result.

\begin{theorem}\label{thm:fmec}
    Consider a consistent family  $ \mathcal{F} $ of subsets for random unitaries  on a graph $(V, E) $ with a finite number of edges $ E $. Let $E \neq   B \in \mathcal{F}  $,  
    $f \in B$, and $ e \in B^c $ and assume:
    \begin{description}
    \item[1. Fractional moments:] for some $ s \in (0,1) $ and $C_{s} < \infty$ and all $ f' \in B$:\begin{equation} \label{eq: assu frac mom}
        \sup_{|z| < 1} \mathbb{E}\left[ |\langle f | (U^B - z)^{-1} f' \rangle|^{s} \right]\leq C_{s}.
        \end{equation}
    \item[2. Spectral averaging:] for some
     $C_W<\infty$ almost surely
\begin{equation} \label{eq: assu spec avg}
   \sup_{|z| < 1}\mathbb E\Big[\langle e \, | {\rm Re} \left( (U+z) (U-z)^{-1} \right) \, e \rangle \, \Big| \, \mathcal{B}(\mathbb{R}^{B})  \Big] \leq C_W ,
\end{equation}
where the expectation conditions on the sigma algebra associated with the edge set~$ B$. 
    \end{description}
   Then for any $ \beta \in (0,s)$ with $ \beta \leq 1 - \beta/s $ and all open $ I \in \mathbb{S}^1$:\begin{equation}\label{eq:ECFFM}
    \mathbb{E}\left[Q(e,f;I,\beta)\right] \leq C_W^{\frac{\beta}{s}} \sum_{f'\in B, e' \in E} t_{f'e'}^\beta \sup_{\delta \in (0,1) } \left\{ \int_I \mathbb{E}\left[ |\langle f | (U^B - \delta e^{i\theta})^{-1} f' \rangle|^{s} \right] \frac{d\theta}{2\pi}\right\}^{\frac{\beta}{s}},
    \end{equation}
    where $ t_{ef} := \sup_\omega |\langle e | 
    T_\omega^{E,B} f \rangle |$. 
\end{theorem}
In a fractional-moment (FM) localization analysis \cite{AizenmanMolchanov:1993}, the theorem is applied to cases where $ E $ is the given finite-volume subset on which the analysis is done. The subset $ B $ is chosen as a ball centered near $ f $ with radius half the distance $ d_E(e,f) $, see Figure~\ref{fig:ECcorr}. Note that the graph distance between two vertices $x,y\in V$ is $d(x,y)$ and as a distance $d_E(e,f)$ between two edges $e,f\in E$ we choose the maximal graph distance between their endpoints. The transition matrix elements $  t_{ef} $ are a worst-case estimate of coupling over the separating surface between $ B $ and its complement within $ E $.  As the difference of two unitaries, we have $ t_{ef} \leq 2 $. In this set-up, the summation in~\eqref{eq:ECFFM} is restricted to $ e' $ and $ f' $ on the boundary of $ B $. In case the FM in the RHS of~\eqref{eq:ECFFM} decays exponentially, this decay has to outweigh the growth of the size of the boundary of $ B $. The application to random SQW is explained in detail in Section~\ref{sec: main proof}. 

Theorem~\ref{thm:fmec} is an adaptation of an analogous result for random self-adjoint operators \cite{AW08}, whose proof is somewhat inspired by \cite{Aizenman06}. 

\begin{figure} [h!]
    \centering
    \includegraphics[scale=0.6]{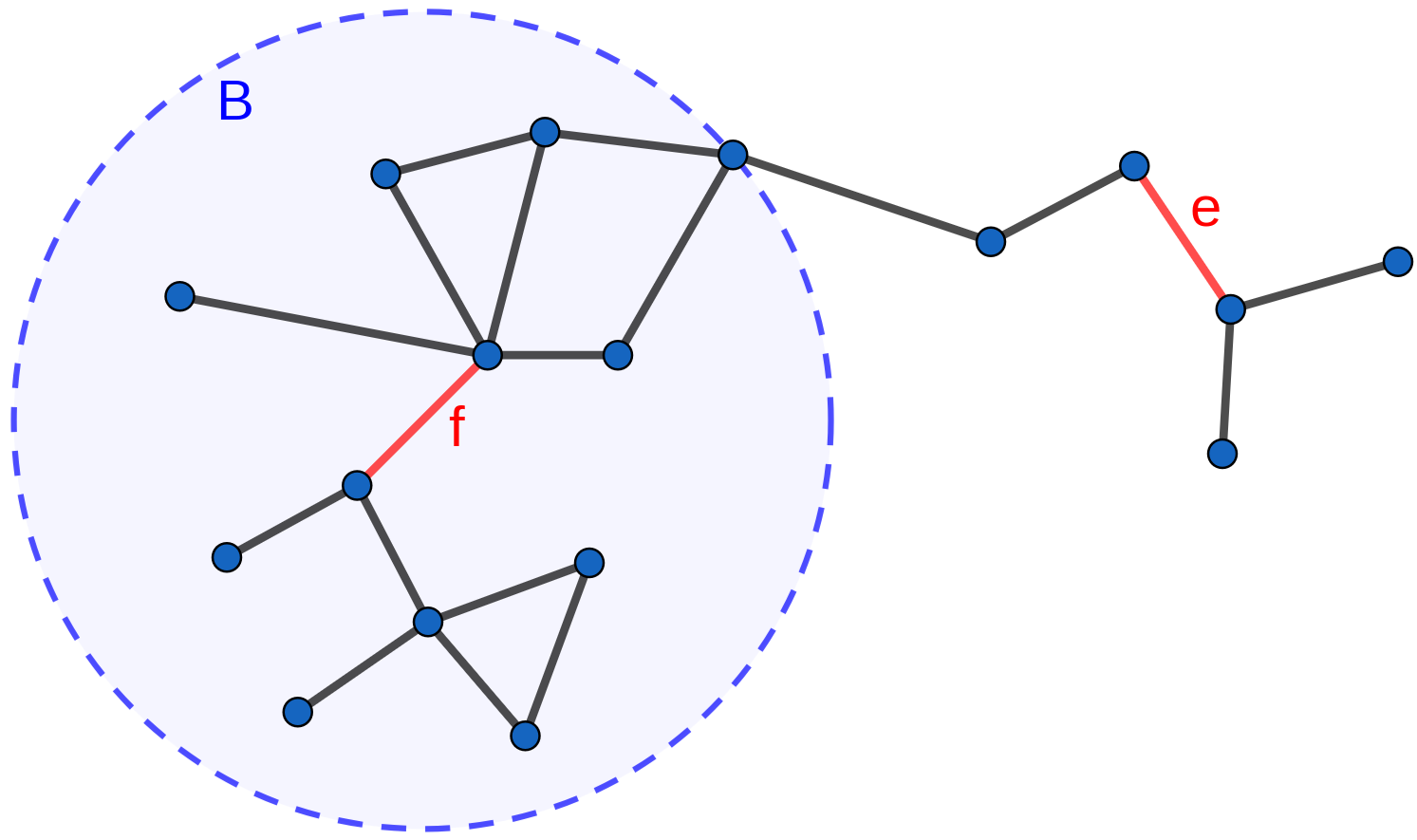}
    \caption{A ball $B$ centered at $f$ of radius $d_E(e,f)/2$.}
    \label{fig:ECcorr}
\end{figure}

\begin{proof}[Proof of Theorem~\ref{thm:fmec}]
    For simplicity, we drop the $\omega$ dependence from the notation in the proof, and we abbreviate the boundary operator by $ T\coloneqq T_\omega^{E,B} = U- U^B \oplus U^{B^c} $.
    Denoting by $P^B$ the projection onto $\ell^2(B)$, we have
\begin{align}
    P^B U = P^B \big( U^B \oplus U^{B^c} + T \big) = P^B U^B \oplus 0 + P^B T \equiv U^B + P^B T.
\end{align}
This allows us to get an expression for $P^BP_\alpha$  with $ P_\alpha $ the eigenprojection of $ U $ corresponding to eigenvalue $\lambda_\alpha\not\in \sigma(U^B)$, see \eqref{decUF}. To avoid dealing with this condition, we introduce a regularization parameter $0<\delta<1$ and write
\begin{align}\label{eq: PB Palpha}
     P^B P_\alpha& =\big( U^B - \lambda_\alpha\delta \big)^{-1} \big((1-\delta) \lambda_\alpha P^B P_\alpha - P^B T  P_\alpha \big)\nonumber\\
     &=-\big( U^B - \lambda_\alpha\delta \big)^{-1}P^B(T-(1-\delta) \lambda_\alpha )P_\alpha = -\big( U^B - \lambda_\alpha\delta \big)^{-1}P^B   T_\alpha^\delta P_\alpha , 
\end{align} 
where we introduced the abbreviation
\begin{align}\label{eq: Tad}
    T_\alpha^\delta\coloneqq T-(1-\delta)  \lambda_\alpha.
\end{align}
Spelling out the projector
$P^B = \sum_{f' \in B} \ket{f'} \bra{f'}$,
and taking  the  $f \in B $, $e\in B^c$  matrix element of
\eqref{eq: PB Palpha}, we conclude
\begin{align}
    \langle f|P_\alpha \,  e\rangle = \sum_{\substack{f' \in B, \, e' \in E}
    } - \langle f \, | \big( U^B - \delta\lambda_\alpha \big)^{-1} \, f' \rangle \, \langle f' \, | T_\alpha^\delta \, e' \rangle \, \langle e' | P_\alpha \,  e\rangle.
\end{align}
Inserting this into the definition~\eqref{def:IEC} and using that $ | a+b|^\beta \leq |a|^\beta + |b|^\beta $ since $ \beta \in (0,1] $, we arrive at
\begin{align}\label{www}
    Q(&e, f; I,\beta) = \sum_{\lambda_\alpha\in I} \langle e| P_\alpha e\rangle^{1-\beta} \, |\langle e| P_\alpha f\rangle|^{\beta} \\ \notag
    &\leq \sum_{\substack{f' \in B, \, e' \in E}} | \langle f' \, | T_\alpha^\delta \, e' \rangle \, |^\beta \sum_{\lambda_\alpha\in I} \langle e| P_\alpha \, e\rangle^{1-\beta}  \,  | \langle f \, | \big( U^B - \delta\lambda_\alpha \big)^{-1} \, f' \rangle|^\beta \,  | \, \langle e' | P_\alpha \,  e\rangle|^\beta.
\end{align}
We pick $p = \frac{s}{\beta} $ and $ q = \frac{s}{s-\beta}$ such that $q^{-1}+p^{-1}=1$, and write $1-\beta=p^{-1}+(1-\beta q)q^{-1}$. An application of  H\"older's inequality yields
\begin{align}
    &Q(e, f; I,\beta) 
    \leq \sum_{\substack{f' \in B, \, e' \in E}} | \langle f' \, | T_\alpha^\delta \, e' \rangle \, |^\beta  \, \Big( \sum_{\lambda_\alpha\in I} \langle e| P_\alpha \, e\rangle  \,  | \langle f \, | \big( U^B - \delta\lambda_\alpha \big)^{-1} \, f' \rangle|^{\beta p} \Big)^\frac{1}{p} \\ &\phantom{{xxxxxxxxxxxxxxxxxxxxxxxxxxxxx} } \times\Big( \sum_{\lambda_\alpha\in I} \langle e| P_\alpha \, e\rangle^{1-\beta q}  \,   | \, \langle e' | P_\alpha \,  e\rangle|^{\beta q} \Big)^\frac{1}{q} \nonumber\\
    &= \sum_{\substack{f' \in B, \, e' \in E}} | \langle f' \, | T_\alpha^\delta \, e' \rangle \, |^\beta   \, Q(e, e', \beta q) ^\frac{1}{q} \, 
    \Big( \sum_{\lambda_\alpha\in I} \langle e| P_\alpha \, e\rangle  \,  | \langle f \, | \big( U^B - \delta\lambda_\alpha \big)^{-1} \, f' \rangle|^{\beta p} \Big)^\frac{1}{p}. \nonumber
\end{align}
Since $\beta q = \frac{\beta s}{s-\beta}\leq 1$, we know $0\leq Q(e,e';I,\beta q) \leq 1$, which yields with $p=\frac{s}{\beta}$
\begin{align} \label{eq: 2nd simone}
    Q(&e, f; I,\beta)  \leq \sum_{\substack{f' \in B, \, e' \in E}} | \langle f' \, | T_\alpha^\delta \, e' \rangle \, |^\beta  \, 
    \Big( \sum_{\lambda_\alpha\in I} \langle e| P_\alpha \, e\rangle  \,  | \langle f \, | \big( U^B - \delta\lambda_\alpha \big)^{-1} \, f' \rangle|^{s} \Big)^\frac{\beta}{s}.
\end{align}
For any $F\in C(\mathbb{S}^1)$ by functional calculus
and using the integral representation as detailed in \cite[Sec. 5.2]{HJS:09}, we have
\begin{align} 
    \sum_{\lambda_\alpha \in I} &F(\lambda_\alpha) \,  \langle e| P_\alpha \, e\rangle = \langle e \, | F(U) P_I(U) \, e \rangle \nonumber \\ &=\lim_{r \to 1^-} \frac{1-r^2}{2 \pi} \, \int_I \langle e \, | (U - re^{i \theta} )^{-1} (U^{-1} - re^{-i \theta} )^{-1}  \, e \rangle F(e^{i \theta} ) d \theta .
\end{align}
Applying this to $F(z) = | \langle f \, | \big( U^B - \delta z \big)^{-1} \, f' \rangle |^{s}$ with $|z\delta|=\delta<1$ yields 
\begin{align} \label{eq: integral rep}
   &\sum_{\lambda_\alpha\in I} \langle e| P_\alpha \, e\rangle  \,  | \langle f \, | \big( U^B - \delta\lambda_\alpha \big)^{-1} \, f' \rangle|^{s}  \\ \nonumber  &= \lim_{r \to 1^-} \frac{1-r^2}{2 \pi} \, \int_I \langle e \, | (U - re^{i \theta} )^{-1} (U^{-1} - re^{-i \theta} )^{-1}\, e \rangle \, | \langle f \, | \big( U^B - \delta e^{i \theta} \big)^{-1} \, f' \rangle|^{s} \, d \theta .
\end{align}
For $|z|\neq 1$, functional calculus yields the identity
\begin{align}
     \text{Re} &\left( (U+z) (U-z)^{-1} \right) = (1-|z|^2) \, (U-z)^{-1} \, (U^{-1} - \bar{z})^{-1} .
\end{align}
Inserting \eqref{eq: integral rep} into \eqref{eq: 2nd simone}, using the computation above as well as the expression~\eqref{eq: Tad} and $0<\beta<1$, we obtain:
\begin{align}
    & Q(e, f; I,\beta)  \leq \sum_{\substack{f' \in B,\,  e' \in E}} \big(  t_{f'e'}^\beta +(1-\delta)^\beta \delta_{e' f'}\big) \,  \, \\ &\times \Big(\lim_{r \to 1^-} \frac{1}{2 \pi} \, \int_I \langle e \, | \text{Re} \left( (U+re^{i\theta}) (U-re^{i\theta})^{-1} \right) \, e \rangle \, | \langle f \, | \big( U^B - \delta e^{i\theta} \big)^{-1} \, f' \rangle |^{s} \, d \theta \Big)^\frac{\beta}{s}.
    \notag
\end{align}
We now take the expectation of the LHS. We then apply Jensen's inequality 
as well as Fatou's and Tonelli's lemmas to get
\begin{align}
    &\E  \big[ Q(e, f; I,\beta) \big] \leq \sum_{\substack{f' \in B, \, e' \in E}} \big(  t_{f'e'}^\beta +(1-\delta)^\beta \delta_{e' f'}\big) \,    \\ & \phantom{x}\times 
    \liminf_{r \to 1^-} \Bigg( \int_I  \E  \Big( \,\langle e \, | \text{Re} \left( (U+re^{i\theta}) (U-re^{i\theta})^{-1}\right) \, e \rangle \nonumber 
    | \langle f \, | \big( U^B - \delta e^{i\theta} \big)^{-1} \, f' \rangle |^{s} \Big)  \frac{d \theta}{2\pi}  \Bigg)^\frac{\beta}{s}.
\end{align}
By assumption the random variable $ \langle f \, | \big( U^B - \delta e^{i\theta} \big)^{-1} \, f' \rangle  $ with $ f , f' \in B $  only depends on the values of $ \omega $ projected to $ B $. Taking the conditional expectation, when conditioning on the latter first, we may use spectral averaging to bound
\begin{align}
    & \E  \Big( \,\langle e \, | \text{Re} \left( (U+re^{i\theta}) (U-re^{i\theta})^{-1}\right) \, e \rangle | \langle f \, | \big( U^B - \delta e^{i\theta} \big)^{-1} \, f' \rangle |^{s} \Big) \notag \\
    &= \E  \Big( \E  \big[\langle e \, | \text{Re} \left( (U+re^{i\theta}) (U-re^{i\theta})^{-1}\right) \, e \rangle | \mathcal{B}(\R^B) \big] \; | \langle f \, | \big( U^B - \delta e^{i\theta} \big)^{-1} \, f' \rangle |^{s} \Big)  \notag \\
    &\leq C_W \ \E  \Big( | \langle f \, | \big( U^B - \delta e^{i\theta} \big)^{-1} \, f' \rangle |^{s} \Big).
\end{align}
Using the finiteness of the fractional moment, we therefore arrive at
\begin{align}
    \E & \big[ Q(e, f; I,\beta) \big] \nonumber \\
    & \leq C_W^{\frac{\beta}{s}}\ \Bigg(\sum_{\substack{f' \in B,\, e' \in E}} |  t_{f'e'}  |^\beta \sup_{\delta\in (0,1)}\Big\{\E \Big[ 
    | \langle f \, | \big( U^B - \delta e^{i\theta} \big)^{-1} \, f' \rangle |^{s}\Big\}^{\frac{\beta}{s}} + C_{s}^{\frac{\beta}{s}}(1-\delta)^{\beta}|B|\Bigg). 
\end{align}
Since the LHS is independent of $\delta$, it remains to take the limit $\delta \to 1^-$ to make the last term of the RHS vanish, which yields the result.
\end{proof}

\section{Application to scattering quantum walks}
\subsection{Scattering quantum walks}\label{sec:SQW}

Following \cite{joye:2024},  scattering quantum walks (SQW) are formulated on the digraph $ (V,D) $ of an underlying graph $(V,E)$ in which the directed edges $ D $ arise from doubling the original edge set and associating with every neighboring pair $ x,y \in V $ in the original graph two directed edges $ (xy) $ and $(yx) $, see Figure~\ref{fig:digraph}.
We write $x \sim y$ whenever there is an edge connecting vertices $x$ and $y$. A directed edge from $x$ to $y$ is denoted by $(yx)$ and $ d_x \coloneqq | \{ y \in V \  \text{ s.t. } \ (yx ) \in D \} | $ counts the number of outgoing edges from $x$, which is the same as the number of incoming edges. \\
\begin{figure} [h!]
    \centering
    \includegraphics[scale=1]{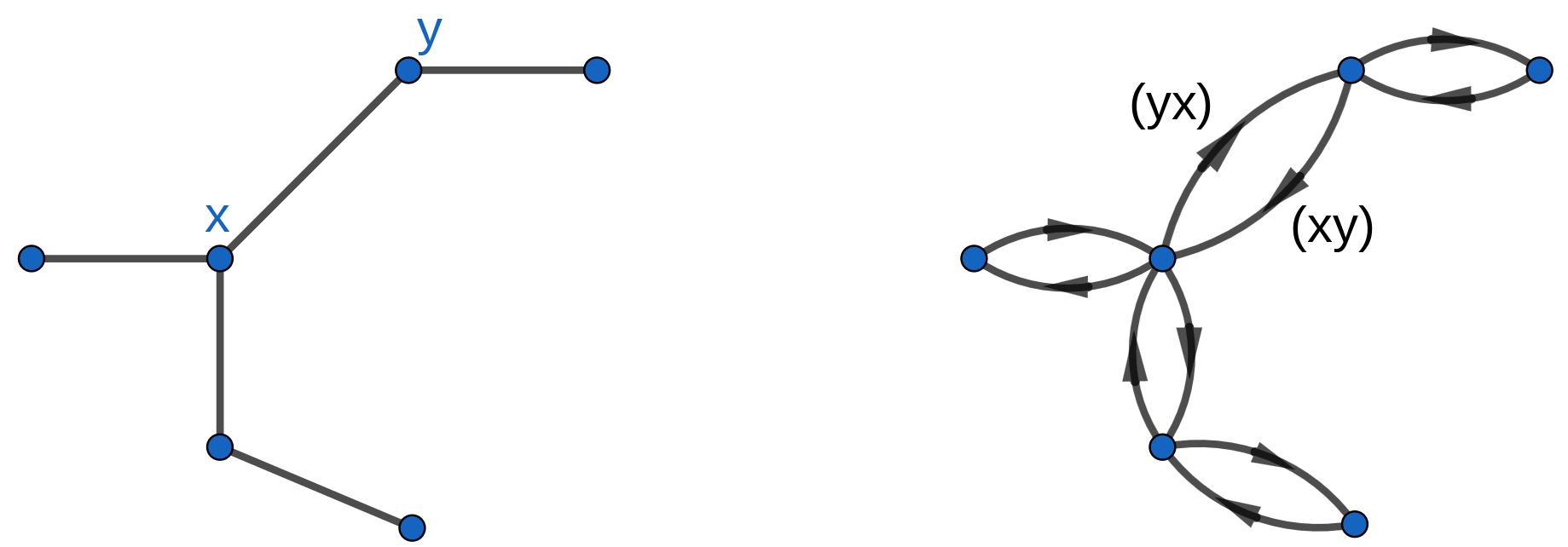}
    \caption{A graph $(V,E)$ and its digraph $(V,D)$}
    \label{fig:digraph}
\end{figure} \\
We will refer to $ d_x $ as the degree of a vertex $x \in V$ and assume throughout that the degree is uniformly bounded
\begin{align} \label{eq: assu bounded degree}
    \exists \,  d <\infty  \text{ s.t. } \forall \, x \in V: \, d_x \leq d .
\end{align}
We define the Hilbert space $\Hp = \ell^2(D)$ and associate to every directed edge $(yx)$ a canonical basis vector $\ket{yx}$.  
We define the subspaces of incoming and outgoing edges at vertex $x \in V$:
\begin{align*}
    \Hp_x^{\rm I} = \text{span} \big\{ \ket{xy} \, \text{ s.t. } y \sim x \big\} \;\; \text{and} \;\; \Hp_x^{\rm O} = \text{span} \big\{ \ket{yx} \, \text{ s.t. } y \sim x \big\}
\end{align*}
and note that
\begin{align*}
    \Hp = \bigoplus_{x \in V} \Hp_x^{\rm I} = \bigoplus_{x \in V} \Hp_x^{\rm O}.
\end{align*}
The orthogonal projections onto $\Hp_x^{\rm I}$ and $\Hp_x^{\rm O}$  are $P_x^{\rm I}$ and $P_x^{\rm O}$ respectively. 

To each vertex, we associate a unitary scattering matrix $S(x) \in \mathcal{U}(d_x)$ and label the entries $ S_{zy}(x) $ by $ z \sim x $ and $ y \sim x $. The family of scattering matrices is denoted by ${\mathcal S}=(S(x))_{x\in V}$. The action of the scattering quantum walk $U_\mathcal{S}$ on an element of the canonical ONB is then defined by
$
    U_\mathcal{S} \ket{xy} = \sum_{z \sim x} S_{zy}(x) \ket{zx} $, see figure \ref{fig:US_graph}.
    \begin{figure}[htbp]
\centering

\begin{tikzpicture}[scale=1.0]


\node[node] (x) at (0,1.5) {$x$};
\node[node] (y) at (-1.8,0.4) {$y$};
\node[node] (z) at (-0.3,-0.8) {$z$};
\node[node] (t) at (1.4,-0.2) {$t$};
\node[node] (u) at (3.2,0.7) {$u$};

\draw[purpleedge,bend left=15] (y) to (x);

\draw[edge,bend left=10] (x) to (y);

\draw[edge,bend left=10] (y) to (z);
\draw[edge,bend left=10] (z) to (y);

\draw[edge,bend left=12] (x) to (z);
\draw[edge,bend left=12] (z) to (x);

\draw[edge,bend left=10] (x) to (u);
\draw[edge,bend left=10] (u) to (x);

\draw[edge,bend left=10] (z) to (t);
\draw[edge,bend left=10] (t) to (z);

\draw[edge,bend left=12] (t) to (u);
\draw[edge,bend left=12] (u) to (t);

\node[color=red] at (-0.3,2.1) {$S(x)$};
\node[color=violet] at (-1.2,1.6) {$|xy\rangle$};


\node at (4,1.5) {$U_{\mathcal S}$};
\draw[->,thick] (3.8,1.1) -- (4.3,1.1);


\node[node] (x2) at (6.6,1.5) {$x$};
\node[node] (y2) at (4.8,0.4) {$y$};
\node[node] (z2) at (6.1,-0.8) {$z$};
\node[node] (t2) at (7.9,-0.2) {$t$};
\node[node] (u2) at (9.8,0.7) {$u$};

\draw[blueedge,bend left=15] (x2) to (y2);
\draw[blueedge,bend left=12] (x2) to (z2);
\draw[blueedge,bend left=12] (x2) to (u2);

\draw[edge,bend left=10] (y2) to (x2);
\draw[edge,bend left=10] (z2) to (x2);

\draw[edge,bend left=10] (u2) to (x2);

\draw[edge,bend left=10] (y2) to (z2);
\draw[edge,bend left=10] (z2) to (y2);

\draw[edge,bend left=10] (z2) to (t2);
\draw[edge,bend left=10] (t2) to (z2);

\draw[edge,bend left=12] (t2) to (u2);
\draw[edge,bend left=12] (u2) to (t2);

\node[color=red] at (6.3,2.1) {$S(x)$};
\node[color=blue] at (8.5,2.1)
{$\sum_{z\sim x} S_{zy}(x)\,|zx\rangle$};

\end{tikzpicture}

\caption{Action of $U_S$}
\label{fig:US_graph}

\end{figure}
In other words, in ket-bra notation we have
\begin{align*}
    U_\mathcal{S} \coloneqq \sum_{x \in V} \sum_{z \sim x} \sum_{y \sim x} S_{zy}(x) \ket{zx} \bra{xy}.
\end{align*}
The operator $U_\mathcal{S}$ is unitary on $ \mathcal{H} $ and, consistent with notions in scattering theory, it intertwines incoming and outgoing subspaces. 
\begin{align}\label{interOI}
    U_\mathcal{S} P_x^{\rm I} = P_x^{\rm O} U_\mathcal{S} .
\end{align} 
 For more details, see \cite{joye:2024}.

\subsection{Random scattering quantum walks}\label{sec:RSQW}
Applications of Theorem~\ref{thm:fmec} to random quantum walks, require enough randomness ensuring the finiteness of fractional moments and spectral averaging. For random SQWs, a rather limited amount of randomness is sufficient. 

Random SQWs are characterized by a deterministic set of scattering matrices  ${\mathcal S}=(S(x))_{x\in V}$ and a probability measure $\mu$ on the torus $\T$, which we assume to be absolutely continuous with bounded density: 
\begin{align}\label{eq: distrib}
d\mu(\theta)=\tau(\theta)d\theta \ \text{with} \ \tau\in L^\infty(\T).
\end{align}
We equip every $x \in V$ with a random variable $\omega_x\in \T$, such that the family $(\omega_x)_{x \in V}$ is independent and identically distributed (i.i.d.) according to $\mu$. We then define the random scattering matrices
    \begin{align}\label{eq: rand1S}
    	S_\omega(x)=e^{i\omega_x}S(x), \ \ \forall x\in V.
    \end{align}
and set the corresponding random operator on $\ell^2(D)$ describing the random SQW as
\begin{align}\label{eq:defrSQW}
	U_\omega \coloneqq U_{{\mathcal S_\omega}}.
\end{align}
By construction, this operator is a unitary of the form $
    U_\omega = D_\omega U_\mathcal{S} $ 
with random diagonal unitary operators acting as
\begin{align} \label{eq: rand1}
 D_\omega \ket{xy} = e^{i \omega_y} \ket{xy}  . 
\end{align}  
This structure is instrumental in the analysis of (de-)localization of random unitary operators, and it is inherited from the construction of the Chalker-Coddington model of condensed matter physics; see the discussions in \cite{Chalker:1988, KOK:2005, ABJ:2010, JM:2010, ABJ:2012}.

    Any random SQW \eqref{eq:defrSQW} with arbitrary underlying deterministic scattering matrices $ {\mathcal S}$ fits into the set-up defined in Section~\ref{sec:consfam}. To do so, one identifies $ (V,E) = (V,D) $ and realizes the randomness on the measure space $(\mathbb{R}^D, \mathcal{B}(\mathbb{R}^D)) $ associated with the edge set instead of the vertex set. This is achieved by setting 
    \begin{equation} \label{eq: randomness on edge}
        \omega_{xy} \coloneqq \omega_y \quad \mbox{for all $ x $ s.t.\ $(xy) \in D$.}
    \end{equation}
    This induces a probability measure, which we again denote by $\mathbb{P}$, on the measure space over the oriented edge set. 
    
    A consistent family $ \mathcal{F} $ of subsets for the random SQW, which fits the requirements of Definition~\ref{ass1}, consists of all subsets $ F \subset D $, which have the property that with an oriented edge $ (xy) \in F $ the reversely oriented edge also belongs to the set. 
    The unitary $ U^F$ from~\eqref{eq:finitevolUF} is then defined as follows. Given $ F \in \mathcal{F} $, its edge boundary of vertices and  the set of inner vertices is 
    \begin{align*}
    V_{\partial F}& \coloneqq \{ x \in V \ \mbox{s.t.} \  \exists y, z \in V : \; (xy) \in F \ \mbox{and} \ \ (zx) \notin F \} \\ 
    V_{F}^\circ & \coloneqq \{ x \in V \ \mbox{s.t.} \  \forall y \sim x : \; (xy) \in F \} .
    \end{align*}
    To define $ U^F$ we modify the scattering matrix on the edge boundary of $ F $ by setting 
    \begin{equation}
        S^F(x) \coloneqq \begin{cases} S(x) & x \in V_{F}^\circ , \\ \mathbbm{1}_{d_x^F}  & x \in V_{\partial F}  . \end{cases} 
    \end{equation}
    Here, the identity matrix in the last case has dimension $ d_x^F \coloneqq | \{ y \in V \ \mbox{s.t.}\ (yx) \in F \} | <d_x$.  
    This family of deterministic unitary scattering matrices  ${\mathcal S}^F=(S^F(x))_{x\in V_{\partial F} \cup V_{F}^\circ }$, which implement "reflecting boundary conditions" on $  V_{\partial F} $, then defines the unitary operator 
    \begin{equation}
        U^F_\omega \coloneqq D_\omega U_{{\mathcal S}^F} , 
    \end{equation}
    which acts on $ \ell^2(F) $ and satisfies all the requirements from~\eqref{eq:finitevolUF}. 
    \medskip

    To control the boundary operator induced by such boundary conditions, $T^{D,F}_\omega$, see \eqref{bdryT}, one may use \cite[Lemma 2.4]{joye:2024}. It guarantees that, 
given two sets of scattering matrices, $\mathcal{S} = \big( S(x) \big)_{x \in V}$ and $\mathcal{S}' = \big( S'(x) \big)_{x \in V}$,  together with the corresponding quantum walks $U_\mathcal{S}$ and $U_{\mathcal{S}'}$,  the natural uniform topology on those can be expressed in terms of the underlying scattering matrices, 
\begin{align} \label{eq: scat mat norm}
    \|U_\mathcal{S} - U_{\mathcal{S}'} \| = \sup_{x \in V} \|S(x) - S'(x) \|_{HS} =: \| \mathcal{S} - \mathcal{S}' \|,
\end{align}
where $\| \cdot \|_{HS}$ denotes the Hilbert-Schmidt norm of a matrix. We will use this in case one of those two sets is the identity on all vertices, $ \mathcal{I} \coloneqq \big( S(x) \big)_{x \in V} $ with $ S(x) = \mathbbm{1} $ for all $ x \in V $. 
In this case, the corresponding random SQW $\tilde{U}_\omega$ satisfies 
\begin{align} \label{eq: I reflection}
    \tilde{U}_\omega \ket{xy}=e^{i \omega_{x}}\ket{yx} 
\end{align}
for all $\ket{xy} \in \Hp $. It is hence trivially fully localized for any randomness. 
Less trivial and 
as an application of Theorem \ref{thm:fmec}, we show dynamical localization whenever the family of scattering matrices $\mathcal{S}$ is close enough to the fully localized case $\mathcal{I}$:
\begin{theorem} \label{thm: dyn loc SQW}
    Let $ U $ be a random SQW  on a digraph $ (V,D) $ with uniformly bounded degree as in~\eqref{eq: assu bounded degree} and random phases that are i.i.d.\ with bounded density as in \eqref{eq: distrib}. Then, for any $g>0$ there exist $\varphi, c > 0$ such that for any family of scattering matrices with $\| \mathcal{S} - \mathcal{I} \| \leq \varphi$, we have for all $\ket{xy}, \ket{x'y'} \in \Hp$ and $I \subset S^1$
    \begin{equation}\label{eq:dynloc}
        \E \Big( \sup_{n \in \Z} | \langle xy \, | U^n \, P_I(U) \, x'y' \rangle | \Big) \leq c \, e^{-g \, d(x,x')} .
    \end{equation}
\end{theorem}
We prove Theorem \ref{thm: dyn loc SQW} using Theorem \ref{thm:fmec}: Sections \ref{sec:fmb} and \ref{sec: spec avg} are devoted to proving the fractional moments and spectral averaging conditions, while we obtain exponential decay of the fractional moments in section \ref{sec: loc proof}. We conclude in Section \ref{sec: main proof}, proving dynamical localization for random SQWs as stated in Theorem \ref{thm: dyn loc SQW}. 
\begin{remark}
    We have equipped every edge $(xy)$ with a random phase $\omega_y$ that depends only on the outgoing vertex $y$, see \eqref{eq: randomness on edge}. If we allow the random phase $\omega_{xy}$ of $(xy)$ to also depend on the incoming vertex $x$, we would not need the new results developed in Section \ref{sec:EC}. Instead, exponential decay of the fractional moments immediately implies dynamical localization using the rank-one analysis in \cite{HJS:09}.
\end{remark}

\subsection{Fractional moment bounds} \label{sec:fmb}
To demonstrate the utility of Theorem~\ref{thm:fmec}, we prove the boundedness of the FM and subsequently a related spectral averaging result in the above set-up of random SQW with any distribution satisfying \eqref{eq: distrib}. 
We formulate both results for the random SQW on an arbitrary digraph. In applications, this digraph may already be a subgraph of a bigger graph. 
The proof of the following bound traces the steps of \cite[Thm. 2]{Schaefer:2025}, which in turn is based on \cite[Thm. 3.1]{HJS:09}. 
\begin{theorem} \label{thm: frac mom bound}
    Given a random SQW $ U $ on a digraph $ (V,D) $ with uniformly bounded degree and $ \mu $ satisfying  \eqref{eq: distrib}, 
    for every $s \in (0,1)$ there exists $C(s) < \infty$ such that
    \begin{align} \label{eq: fract mom bound}
      \sup_{|z| \neq 1} \int_\T \int_\T | \langle yx \, | (U_\omega - z )^{-1} y'x' \rangle |^s \; d\mu( \omega_{x}) \, d\mu(\omega_{x'}) \leq C(s)
    \end{align}
    for all $(yx), (y'x') \in D$ and arbitrary values of all other random variables $\omega_{v}$, $v\notin \{x, x'\}$. 
\end{theorem}
\begin{proof}
    The bound \eqref{eq: fract mom bound} is trivial when $z$ is away from the unit circle:
    \begin{align*}
        | \langle yx \, | (U_\omega - z )^{-1} \, y'x' \rangle | \leq \|(U_\omega - z )^{-1} \| = \frac{1}{\text{dist}(z, \sigma(U_\omega))} \leq \frac{1}{| 1 - |z| \, |}. 
    \end{align*}
    We may thus assume without loss of generality, that $1/2 \leq |z| \leq 3/2$ and $|z| \neq 1$. Furthermore, for $z\neq 0$ it holds that
    \begin{align}\label{eq:zneq0}
        \frac{1}{2z} \, \Big( (U_\omega + z) (U_\omega - z)^{-1} - 1 \Big) = (U_\omega - z)^{-1}.
    \end{align}
    To prove \eqref{eq: fract mom bound} it thus suffices to show the similar bound on $(U_\omega + z)(U_\omega-z)^{-1}$:
    \begin{align} \label{eq: bound on U plus z U minus z}
        \int_\T \int_\T | \langle yx \, | (U_\omega+z) \, (U_\omega - z )^{-1} \, y'x' \rangle |^s \; d\mu( \omega_{x}) \, d\mu(\omega_{x'}) \leq C(s).
    \end{align}
    We assume for now that $x \neq x'$, while the simpler case $x = x'$ will be dealt with at the end. Following the steps of \cite{Schaefer:2025} and \cite{HJS:09}, we introduce the change of variables
    \begin{align*}
        \alpha = \frac{1}{2} \big( \omega_{x} + \omega_{x'} \big) \;\;\; \text{and} \;\;\; \beta = \frac{1}{2} \big( \omega_{x} - \omega_{x'} \big),
    \end{align*}
    and the unitary operators:
    \begin{align*}
        D_\alpha \ket{vu} &= \begin{cases}
            e^{i \alpha} \, \ket{vu} & \text{if } u \in \{ x, x' \} \\
            \ket{vu} & \text{else}
        \end{cases}, \;\;
        D_\beta \ket{vu} = \begin{cases}
            e^{i \beta} \, \ket{vu} & \text{if } u = x\\
            e^{-i \beta} \, \ket{vu} & \text{if } u = x' \\
            \ket{vu} & \text{else}
        \end{cases} \\ \widehat{D}_\omega \ket{vu} &= \begin{cases}
            \ket{vu} & \text{if } u \in \{ x, x' \} \\
            e^{i \omega_{u}} \, \ket{vu} & \text{else}
        \end{cases}.
    \end{align*}
    By definition, we have $D_\omega = D_\alpha \, D_\beta \, \widehat{D}_\omega$. We define the unitary operator $V_\omega = D_\beta \, \widehat{D} \, U_\mathcal{S}$, and denote by $P$ the projection onto the subspace
    \begin{align}\label{Pproj}
        \text{Ran}(P) = \text{span} \left\{ V_\omega^{-1} \ket{vx}, V_\omega^{-1} \ket{v'x'}, \, v\sim x, v' \sim x' \right\}.
    \end{align}
    We easily verify that $(U_\omega - V_\omega) \, (1 - P) = 0$ and $V_\omega^{-1} \, U_\omega P = e^{i \alpha} P$ and therefore conclude
    \begin{align} \label{eq: fmb estimate 1}
        U_\omega = U_\omega \, P + U_\omega \, (1-P) = e^{i \alpha} \, V_\omega \, P + V_\omega \, (1-P).
    \end{align}
  For $|z| \neq 1$,  we define the operators
    \begin{align}\label{defFF}
        F_z = P (U_\omega + z) (U_\omega -z)^{-1} P \text{  and  } \widehat{F}_z = P (V_\omega + z) (V_\omega -z)^{-1} P,
    \end{align}
    viewing them as operators on the range of $P$. Using that $[V_\omega +z, (V_\omega -z)^{-1}] =0$ and $\big( (V_\omega - z)^{-1} \big)^* = (V_\omega^* - \overline{z})^{-1}$, we compute
    \begin{align}\label{FpFstar}
    \begin{split}
        \widehat{F}_z + \widehat{F}_z^* = P \Big( 2\big( I - |z|^2 \big) (V_\omega -z)^{-1} \, \big( (V_\omega -z)^{-1} \big)^* \Big) P.
    \end{split}
    \end{align}
    Assuming that $|z| > 1$, we conclude that $\widehat{F}_z + \widehat{F}_z^* < 0 $ is invertible. Thus, $-i \widehat{F}_z$ is a dissipative operator if $|z| > 1$, where we recall that an operator $A$ is dissipative iff its imaginary part $\frac{1}{2i} (A-A^*)$ is positive. Using $\widehat{F}_z^{-1} = \widehat{F}_{-z}$, we conclude that $-i \widehat{F}_z^{-1}$ is also dissipative for $|z| > 1$. Similarly, we obtain that $i \widehat{F}_z$ and $i \widehat{F}_z^{-1}$ are dissipative for $|z| < 1$. Employing the resolvent identity and \eqref{eq: fmb estimate 1}, we compute
    \begin{align*}
        F_z - \widehat{F}_z = P \Big( 2z (V_\omega - z)^{-1} (e^{i \alpha}-1) V_\omega \, P \, (U_\omega -z)^{-1} \Big) P.
    \end{align*}
    We use the definitions of $F_z$ and $\widehat{F}_z$ to obtain on the range of $P$:
    \begin{align*}
        F_z - \widehat{F}_z = \frac{1}{2} \, (e^{i \alpha} -1) \, (1 + \widehat{F}_z) \, (1 - F_z). 
    \end{align*}
    Following the steps from \cite{HJS:09}, we define $m(\alpha) = i \frac{1+e^{i \alpha}}{1-e^{i \alpha}} = - \cot \big(\frac{\alpha}{2}\big)$ for $\alpha \notin \pi \Z$ and obtain from a direct computation:
    \begin{align} \label{eq: fmb estimate 2}
        F_z = -i \Big( -i \widehat{F}_z + m(\alpha) \Big)^{-1} -i \Big( -i \widehat{F}_z^{-1} - \frac{1}{m(\alpha)} \Big)^{-1}.
    \end{align}
    Using $V_\omega^{-1} \, U_\omega P = e^{i \alpha} P$ and unitarity of $V_\omega$, we obtain:
    \begin{align*}
        \langle yx \, | \, (U_\omega+z) \, (U_\omega -z)^{-1} \, y'x' \rangle = \langle V_\omega^* \, yx \, | \, F_z V_\omega^* \, y'x' \rangle.
    \end{align*}
    This allows us to estimate the integral in \eqref{eq: bound on U plus z U minus z}:
    \begin{align*}
        &\int_{\T^2} | \langle yx \, | \, (U_\omega+z) \, (U_\omega - z )^{-1} y'x' \rangle |^s \; d\mu( \omega_{x}) \, d\mu(\omega_{x'}) \leq \| \tau \|_\infty^2 \int_{-\pi}^\pi \int_0^{2\pi} \| F_z \|^s \; d\alpha \, d\beta .
    \end{align*}
    We stress that while $F_z$ depends on $\alpha$, $\widehat{F}_z$ does not. Note that for $s < 1$, we have $|a+b|^s \leq |a|^s + |b|^s$ for all $a,b \in \R$. Plugging equation \eqref{eq: fmb estimate 2} into the integral above, we obtain
    \begin{align} \label{eq: fmb estimate 3}
        &\int_0^{2\pi} \| F_z \|^s \; d\alpha 
        \leq \int_0^{2\pi} \| \big(-i \widehat{F}_z + m(\alpha) \big)^{-1} \|^s \; d\alpha + \int_0^{2\pi} \| \Big(-i \widehat{F}_z^{-1} - \frac{1}{m(\alpha)} \Big)^{-1} \|^s \; d\alpha .
    \end{align}
    Using that $\widehat{F}_z$ does not depend on $\alpha$, we substitute $t=m(\alpha)$ to bound the first integral on the RHS of \eqref{eq: fmb estimate 3}:
    \begin{align*}
        \int_0^{2\pi} \| \big(-i \widehat{F}_z + m(\alpha) \big)^{-1} \|^s \; d\alpha &= \lim_{R \to \infty} \int_{-R}^{R} \frac{2}{1+t^2} \, \| \big(-i \widehat{F}_z + t \big)^{-1} \|^s \; dt \\ &\leq 2 \sum_{n \in \Z} \frac{1}{1+(|n|-1)^2} \, \int_n^{n+1} \| (-i \widehat{F}_z + t)^{-1} \|^s \, dt.
    \end{align*}
    A similar expression holds for the second term in \eqref{eq: fmb estimate 3}. The operators $-i \widehat{F}_z$ and $-i \widehat{F}^{-1}_z$ are defined on $\text{Ran}(P)$, which has dimension $d_x + d_{x'}$ and are thus maximally dissipative. Reproducing the arguments from the end of Section 4.1 in \cite{Schaefer:2025}, we get that there exists a constant $C(s)$ such that
    \begin{align*}
        \int_n^{n+1} \| (A + t)^{-1} \|^s \, dt \leq C(s)
    \end{align*}
    for all maximally dissipative $A$, where $C(s)$ is independent of $n$, but depends on the dimension of $\text{Ran}(P)$. Since such a constant exists for all finite dimensions, and the maximal dimension is $2d$, we can take the maximum over all such constants. Using that $\sum_{n \in \Z} \frac{1}{1 + (|n|-1)^2} < \infty$, we can therefore bound the RHS of \eqref{eq: fmb estimate 3} for all $|z| > 1$. If $|z| < 1$, we drop a minus sign in both norms in \eqref{eq: fmb estimate 3} and use that $i \widehat{F}_z$ and $i \widehat{F}_z^{-1}$ are dissipative. Repeating the arguments from above yields the bound for all $|z| < 1$.

    The case $x = x'$ is significantly easier: We can directly take $\alpha = \omega_{x}$ and do not need $\beta$. The operators $F_z$ and $\widehat{F}_z$ act on a space of dimension $d_x$, while all other estimates still hold.
\end{proof}

\subsection{Spectral averaging} \label{sec: spec avg}

For random SQW, spectral-averaging over a single random variable at vertex $ x $ ensures the bound required in Theorem~\ref{thm:fmec}. We again only require the distribution to be absolutely continuous with bounded density~\eqref{eq: distrib}. Thanks to~\eqref{eq: bound on U plus z U minus z}, the subsequent bound \eqref{eq: assu integral bound} would hold trivially if the  integrand were raised to some power $s\in (0,1)$, see Theorem \ref{thm: frac mom bound} and \eqref{eq: bound on U plus z U minus z}.

\begin{theorem} \label{thm: spec avg}
Given a random SQW $ U_\omega $ on a digraph $ (V,D) $ with uniformly bounded degree  and $ \mu $ satisfying  \eqref{eq: distrib}, 
there is some $ C < \infty $ such that for all $  (x'x) \in D $  
\begin{align} \label{eq: assu integral bound}
  \sup_{|z| < 1 }  \int_{\mathbb{T}} \langle x'x \, | \text{Re} \left( (U_\omega+z) (U_\omega-z)^{-1} \right) \, x'x \rangle \, d \mu(\omega_x) \leq C ,
\end{align}
at arbitrary values of all other random variables $\omega_{v}$, $v\notin \{x\}$. 
\end{theorem}
\begin{proof}
The proof is based on similar arguments as the diagonal case in Theorem~\ref{thm: frac mom bound}. To illuminate the similarity, we abbreviate $\alpha = \omega_x \in [0, 2\pi]$ for fixed $x \in V $. We write the random operator $U_\omega = D_\omega U_\mathcal{S} $ and define
$$
   D_\alpha \ket{vu} = \begin{cases}
            e^{i \alpha} \, \ket{vu} & \text{if } u \in \{ x \} \\
            \ket{vu} & \text{else}
        \end{cases}, \qquad \widehat{D}_\omega \ket{vu} = \begin{cases}
            \ket{vu} & \text{if } u \in \{ x \} \\
            e^{i \omega_{u}} \, \ket{vu} & \text{else}
        \end{cases}.
$$
Setting $V_\omega = \widehat{D}_\omega$, we have $U_\omega = D_\alpha V_\omega$. We again consider the orthogonal projection $P $ corresponding to $\text{Ran}(P) = \text{span} \left\{ V_\omega^{-1} \ket{vx},  \, v\sim x\right\}$, 
and the operators, 
\begin{align}
    F_z = P (U_\omega +z) (U_\omega -z)^{-1} P, \qquad
    \widehat{F}_z &= P (V_\omega+z) (V_\omega -z)^{-1} P .
\end{align}
We stress that $\widehat{F}_z$ is independent of $\alpha$. 
For these operators, the relation~\eqref{FpFstar} is valid, which implies that 
$\widehat{F}_z + \widehat{F}_z^* > 0$ for any $|z|<1$. By the same argument as in the proof of Theorem~\ref{thm: frac mom bound}, this implies that  $\widehat{F}_z $ and $ \widehat{F}_z^*$ are invertible on the range of $P$. 
Moreover, we also have
\begin{align} \label{eq: real part rewriting}
   0 \leq  \langle x'x \, | \text{Re} \left( (U+z) (U-z)^{-1} \right) \, x'x \rangle = \frac{1}{2} \, \langle x'x \, | \, \big( F_z + F_z^* \big) \, x'x \rangle .
\end{align}
The operators $ F_z $ and $ \widehat{F}_z $ are again related by~\eqref{eq: fmb estimate 2} 
with $m(\alpha) = - \cot(\frac{\alpha}{2})$ such that
\begin{align}
    F_z + F_z^* 
    &= \Big(  \widehat{F}_z + i m(\alpha) \Big)^{-1} + \Big( \widehat{F}_z^{-1} - i m(\alpha)^{-1} \Big)^{-1} + \Big( \widehat{F}_z^* -i m(\alpha) \Big)^{-1} 
    +  \Big( (\widehat{F}_z^*)^{-1} +i m(\alpha)^{-1} \Big)^{-1}
\end{align}
By the assumption~\eqref{eq: distrib} on the distribution, the claim~\eqref{eq: assu integral bound}  thus follows from the bounds
\begin{equation} \label{eq: assu integral bound 2}
\begin{split}
    \int_0^{2 \pi} \langle x'x \, | \, \Big(  \widehat{F}_z + i m(\alpha) \Big)^{-1} + \Big( \widehat{F}_z^* -i m(\alpha) \Big)^{-1} \, x'x \rangle \, d \alpha \leq C \\
    \int_0^{2 \pi} \langle x'x \, | \, \Big( \widehat{F}_z^{-1} - i m(\alpha)^{-1} \Big)^{-1} + \Big( (\widehat{F}_z^*)^{-1} +i m(\alpha)^{-1} \Big)^{-1} \, x'x \rangle \, d \alpha \leq C 
\end{split}
\end{equation}
that we now establish. 
To simplify the notation, we will write $\widehat{F}_z = \widehat{F}$ and $m(\alpha) = m$, and we stress that from now on, we consider the restriction of all operators to the finite-dimensional range of $P$, without systematically mentioning it explicitly in the notation. 
Since the operator $\widehat{F}$ restricted to the range of $P$ is dissipative, we have $H = \frac{1}{2}(\widehat{F}+\widehat{F}^*) = H^* > 0$ and 
\begin{align}
    r = \inf_{\substack{\varphi \in \text{Ran}(P) \\ \| \varphi \| = 1} }\langle \varphi \, | H \varphi \rangle = \min \big( \lambda \in \sigma(H) \big) > 0 . 
\end{align}
We now pick $\varphi \in \text{Ran}(P)$ arbitrary and  compute for  $y(t) = e^{-t\widehat{F}} \varphi$ the derivative
\begin{align} \notag 
    \frac{d}{dt} \| y(t) \|^2 &= \Big\langle \varphi \, | \, \frac{d}{dt} \Big( e^{-t\widehat{F}^*} \, e^{-t\widehat{F}} \Big) \, \varphi \Big\rangle = \langle \varphi \, | \, - e^{-t\widehat{F}^*} (\widehat{F}^* + \widehat{F}) e^{-t\widehat{F}} \, \varphi \rangle \\
    &= -2 \, \langle y(t) \, | H \, y(t) \rangle \leq -2 \, r \, \|y(t)\|^2 .
\end{align}
Gronwall's inequality then yields
$    \| y(t) \|^2 \leq \| \varphi \|^2 e^{-2rt} $, 
which proves the uniform bound
\begin{equation}\label{eq: bound on norm of expo}
  \| e^{-t \widehat{F}} \| \leq e^{- t r}  \quad\mbox{on $\text{Ran}(P)$.}
  \end{equation}
  We can therefore compute
\begin{align}
    \int_0^\infty e^{-t(\widehat{F}+im)} dt &= \int_0^\infty \frac{d}{dt} \Big( -(\widehat{F}+im)^{-1} \, e^{-t(\widehat{F}+im)} \Big) \, dt \\ \nonumber &= \lim_{t \to \infty} \Big( - (\widehat{F}+im)^{-1} \, e^{-t(\widehat{F}+im)} \Big) + (\widehat{F}+im)^{-1} = (\widehat{F}+im)^{-1} 
\end{align}
and similarly 
$ 
    \int_0^\infty e^{-t(\widehat{F}^* - im)} dt = (\widehat{F}^* - im)^{-1} $. 
Restoring the dependence in $\alpha$ in the notation, and using Fubini's theorem we thus conclude
\begin{align*}
    &\int_0^{2 \pi} \langle x'x \, | \, \big( \widehat{F}+ i m(\alpha) \big)^{-1} \, x'x \rangle \, d \alpha = \int_0^{2 \pi} \langle x'x \, | \, \Big( \int_0^\infty e^{-t(\widehat{F}+im(\alpha))} dt \Big) \, x'x \rangle \, d \alpha \\
    &= \int_0^{2\pi} \int_0^\infty \langle x'x \, | \, e^{-t(\widehat{F}+im(\alpha))} \, x'x \rangle \, dt \, d \alpha = \int_0^\infty \langle x'x \, | \, e^{-t\widehat{F}} \, \Big( \int_0^{2\pi} e^{-itm(\alpha)} \, d \alpha \Big) \, x'x \rangle \, dt .
\end{align*}
The substitution $y = m(\alpha) = -\cot(\frac{\alpha}{2})$ with $d \alpha  = \frac{2}{1+y^2} dy $ yields for the inner integral
\begin{align*}
    \int_0^{2\pi} e^{-itm(\alpha)} \, d \alpha = 2\int_{-\infty}^\infty e^{-tiy} \frac{1}{1+y^2} \, dy = 2 \pi \, e^{-|t|} .
\end{align*}
The final equality can be checked with the residue theorem. We substitute this into the equation above and obtain
\begin{align*}
    &\int_0^{2 \pi} \langle x'x \, | \, \big( \widehat{F}+ i m(\alpha) \big)^{-1} \, x'x \rangle \, d \alpha = 2 \pi \, \int_0^\infty \langle x'x \, | \,  e^{-t(\widehat{F}+1)} \, x'x \rangle \, dt \\ 
    &= 2\pi \langle x'x \, | \Big( \int_0^\infty e^{-t(\widehat{F}+1)} \, dt \Big) x'x \rangle = 2 \pi \, \langle x'x \, | \, (\widehat{F}+1)^{-1} \, x'x \rangle ,
\end{align*}
by using the argument above in the evaluation of the last integral. Taking the absolute value, we can bound the scalar product by $\| (\widehat{F}+1)^{-1} \|$. Due to $\widehat{F}+\widehat{F}^* > 0$ we compute for every $\lambda > 0$ and $0\neq \varphi\in \text{Ran}(P)$
\begin{align*}
    \| ( \lambda +\widehat{F} ) \varphi \|^2 = \lambda^2 \| \varphi \|^2 + \| \widehat{F} \varphi \|^2 + \lambda \langle \varphi \, | (\widehat{F} + \widehat{F}^*) \varphi \rangle >  \lambda^2 \| \varphi\|^2 .
\end{align*}
We see that $\| (\lambda + \widehat{F}) \varphi \| > \lambda \| \varphi \|$ and conclude that $\lambda + \widehat{F}$ is invertible on $\text{Ran}(P)$ with
\begin{align*}
    \| (\lambda + \widehat{F})^{-1} \varphi \| < \frac{1}{\lambda} \, \| (\lambda + \widehat{F}) ( \lambda+\widehat{F})^{-1} \varphi \| = \frac{1}{\lambda} \, \| \varphi \| .
\end{align*}
In particular, we have $\| (\widehat{F}+1)^{-1} \| \leq 1$. In summary, we have shown
\begin{align*}
    \Big| \int_0^{2 \pi} \langle x'x \, | \, \big( \widehat{{F}}_z+ i m(\alpha) \big)^{-1} \, x'x \rangle \, d \alpha \Big| \leq 2 \pi .
\end{align*}
Copying the arguments above, we obtain the same bound with $ \widehat{{F}}_z $ replaced by $ \widehat{{F}}_z^*$.
Using $\widehat{F}_z^{-1} + ( \widehat{F}_z^{-1} )^* = \widehat{F}_z^{-1} \big( \widehat{F}_z + \widehat{F}_z^* \big) (\widehat{F}_z^{-1})^*$, we conclude $ \widehat{F}_z^{-1} + (\widehat{F}_z^{-1})^* > 0$. Since $\widehat{F}_z + \widehat{F}_z^* > 0$ was the only requirement on $\widehat{F}_z$, we can repeat the arguments above and obtain similar bounds for $\widehat{F}_z^{-1} $ and $(\widehat{F}_z^{-1})^*$. This shows the bounds \eqref{eq: assu integral bound 2} and therefore \eqref{eq: assu integral bound}, which finishes the proof.
\end{proof}

\section{Localization proof} \label{sec: loc proof}
In order to obtain dynamical localization for random SQWs, it remains to show exponential decay of the fractional moments. Thanks to Theorem \ref{thm:fmec}, this yields exponential decay of the eigenfunction correlator, which proves dynamical localization by~\eqref{eq: loc from ec decay}. We show the following
\begin{theorem} \label{thm: expo dec frac mom}
    Under the assumptions of Theorem \ref{thm: dyn loc SQW}, for any $s \in (0, \frac{1}{3})$ and any $g>0$, there exist $\varphi, c >0$ such that for any collection of scattering matrices $(S(x))_{x \in V}$ with $\| \mathcal{S} - \mathcal{I} \| \leq \varphi$ we have
    \begin{align} \label{eq: expo dec of frac mom}
      \sup_{|z|\in [1/2,3/2]\setminus \{1\}}  \E \left( | \langle xy | (U - z)^{-1} | x'y' \rangle |^s \right) \leq c e^{-g \dist(x,x')},
    \end{align}
    for any two basis vectors $\ket{xy}, \ket{x'y'} \in \ell^2(D)$.
\end{theorem}

\subsection{Subspaces associated to balls} \label{subsec: restricting to balls}
Given a root $r \in V$ and a radius $n$, we will work with a geometric resolvent equation, which imposes reflecting boundary conditions on the surface of the ball $B_n(r)\subset V$. The quantum walker is hence trapped inside or outside this ball. For practical purposes, we apply the ideas presented in Section \ref{sec:RSQW} to subsets of vertices instead of subsets of edges. Given a root $r \in V$, which we frequently suppress in the notation, we denote by $\dist(x) = d(x,r)$ the graph distance between $x$ and $r$ and define
\begin{align*}
    B_n \equiv B_n(r) = \{ x \in V \, | \dist(x) \leq n \} \text{ and } S_n(r) = \{ x \in V \, | \dist(x) = n \}. 
\end{align*}
Due to assumption \eqref{eq: assu bounded degree}, we can bound the size of $S_n(r)$ by
\begin{align} \label{eq: def sn}
    s_n = \sup_{r \in V} \big( |S_n(r)| \big) \leq d \, (d-1)^{n-1} = c_d \, e^{\gamma n} \text{ with } \gamma = \ln(d-1).
\end{align}
For a given family of scattering matrices $(S(x))_{x \in V}$, we define a new family of deterministic scattering matrices by:
\begin{align*}
    S^{(n)}(x) = \begin{cases} \mathbbm{1} &\text{if } x \in S_n(r) \\
    S(x) &\text{else.} \end{cases}
\end{align*}
Let $U_\omega^{(n)}$ denote the random SQW  corresponding to the scattering matrices $(S^{(n)}(x))_{x \in V}$. Recalling \eqref{eq: I reflection}, we see that the walker incoming at $x\in S_n(r)$ from any site $y$, gets reflected from $x$ back to $y$. We define the subspaces
\begin{align*}
    \Hp_n = \text{span} \left\{ \ket{xy} \text{ s.t. } \dist(x) \leq n, \dist(y) \leq n,\, x\sim y,\, (x,y)\in V^2 \right\}
\end{align*}
and $\Hp_n^\perp$ its orthogonal complement in $\Hp$. Thanks to \eqref{eq: I reflection}, the spaces $\Hp_n$ and $\Hp_n^\perp$ are invariant under $U_\omega^{(n)}$, and we let $U_\omega^{B_n}$ and $U_\omega^{B_n^c}$ denote the restrictions of $U_\omega^{(n)}$ to $\Hp_n$ and $\Hp_n^\perp$ respectively, so that
\begin{align}
	U_\omega^{(n)}=U_\omega^{B_n}\oplus U_\omega^{B_n^c}  \text{ on } \Hp=\Hp_n\oplus \Hp_n^\perp.
\end{align}
We denote the boundary operator that couples the two subspaces by
\begin{align*}
    T^{(n)} = U_\omega - U_\omega^{(n)}.
\end{align*}
Using \eqref{eq: scat mat norm}, we obtain 
 $\| T^{(n)} \| \leq \sup_{x \in S_n(r)} \big( \| S(x) - I \|_{HS} \big).$

\subsection{Spectral gaps of the fully localized random operator} \label{sec: spec gaps full loc op}
Let $\Tilde{U}_\omega$ denote the random SQW corresponding to the family of scattering matrices $\mathcal{I}$ as defined above \eqref{eq: I reflection} and note that it reflects any directed edge $(xy)$ onto the edge $(yx)$ in the opposite direction. Given two adjacent vertices $u \sim v$, we define the 2-dimensional subspace
$ \Hp^{uv} = \text{span} \big\{ \ket{uv}, \ket{vu} \big\} $. 
The subspaces $\Hp^{uv}$ are invariant under $\tilde{U}_\omega$ and its restriction $\tilde{U}_\omega^{uv}$ in the ONB $\{\ket{uv}, \ket{vu}\}$ of $\Hp^{uv}$ reads:
\begin{align}\label{uuv}
    \tilde{U}_\omega^{uv} = \tilde{U}_\omega \Big|_{\Hp^{uv}} = \begin{pmatrix}
    0 & e^{i \omega_u} \\ e^{i \omega_v} & 0 \end{pmatrix}.
\end{align}
We see that its eigenvalues are given by $\pm e^{i \frac{\omega_u + \omega_v}{2}}$. Our next goal is the following Proposition, which estimates the probability of spectral gaps for the fully localized operator: 
\begin{proposition} \label{propo: spect gaps fully loc}
    Under the assumptions of Theorem \ref{thm: dyn loc SQW}, we have for any $n \in \N^*$, any $0<\eta <1$ and any root $r \in V$:
    \begin{align*}
      \sup_{|z| \neq 1 }  \Pp \Big( \text{dist} \big(z, \sigma \left( \tilde{U}_\omega^{B_n(r)} \right) \big) \leq \eta \Big) \leq 4 \pi^2 \| \tau \|_\infty^2 \, d |B_n(r)| \, \eta .
    \end{align*}
\end{proposition}
\begin{proof}
    Since $\tilde{U}_\omega^{B_n(r)}=\bigoplus_{ \substack{u,v \in B_n(r) \\ u\sim v}} \tilde{U}^{uv}_\omega$, with $\sigma(\tilde{U}^{uv}_\omega)=\{\pm e^{i\frac{\omega_u+\omega_v}{2}}\}$, see \eqref{uuv}, we have
    \begin{align}
        \Pp \Big( \text{dist} \big(z, \sigma \left( \tilde{U}_\omega^{B_n(r)} \right) \big) \leq \eta \Big)&\leq \sum_{\substack{u,v \in B_n(r) \\ u\sim v}} \Pp \Big( \text{dist} \big(z, \sigma \left( \tilde{U}_\omega^{uv} \right) \big) \leq \eta \Big) \nonumber \\ & \leq \sum_{\substack{u,v \in B_n \\ u \sim v}} \Pp (|z-e^{i\frac{\omega_u+\omega_v}{2}}|\leq \eta)+\Pp (|z+e^{i\frac{\omega_u+\omega_v}{2}}|\leq \eta). 
    \end{align}
In turn, we estimate
\begin{align}
    \Pp \Big( |z\pm e^{i\frac{\omega_u+\omega_v}{2}}|\leq \eta \Big)\leq \| \tau \|_\infty^2 \, \int_{0}^{2\pi}\int_{0}^{2\pi} \mathbb{I}_{\{|z\pm e^{i\frac{\theta+\theta'}{2}}|\leq \eta\}}d\theta \, d\theta', 
\end{align}
and then change variables $\alpha=\frac{\theta+\theta'}{2}$, $\beta=\theta-\theta'$. The bound
\begin{align*}
    \int_{0}^{2\pi} \mathbb{I}_{\{|z\pm e^{i\alpha}|\leq \eta\}}d\alpha \leq 2 \arcsin(\eta) \leq \pi \, \eta 
\end{align*}
completes the proof.
\end{proof}

\subsection{Finite-volume resolvent estimate} \label{sec: poly dec box res}
We denote by $R_\omega^{\#}(z)$ and $\tilde{R}_\omega^{\#}(z)$ the (reduced) resolvents of $U_\omega^{\#}$ and $\tilde{U}_\omega^{\#}$, for $\#\in\{(n), B_n, B_n^c\}$, for any $z$ in their resolvent sets. Unless it is unclear, we will omit the dependence on $z$ and $ \omega$. Furthermore, we define the following equivalence relation on the elements of the ONB $\{ \ket{uv}\}_{u,v \in V, u \sim v}$ of $\Hp$:
\begin{align*}
    \ket{uv} \sim \ket{u'v'} \iff \ket{uv} = \ket{u'v'} \text{ or } \ket{uv} = \ket{v'u'}.
\end{align*}
Recalling \eqref{eq: scat mat norm}, our goal is the following
\begin{proposition} \label{propo: poly decay resolvent}
    Under the assumptions of Theorem \ref{thm: dyn loc SQW}, we have for any $s \in (0,1)$ and $p > \frac{1}{1-s}$ that there exists $c > 0$ (depending on $d, s, p$) such that:
    \begin{align*}
        \E \Big( \left| \langle xy \, | \, R_\omega^{B_n}(z) \, x'y' \rangle \right|^s \Big) \leq c\left(\|\mathcal S-\mathcal I\||B_n|^2\right)^\frac{s}{1+2sp},
    \end{align*}
    uniformly in $z \in \C \setminus \Sp^1$, for all $n \in \N^*$, all balls $B_n \equiv B_n(r) $ for any root $r \in V$, all families of scattering matrices $\mathcal{S} = (S(x))_{x \in V}$ such that $\| \mathcal{S} - \mathcal{I} \| < |B_n|^{-2} d^{-\frac{1+2ps}{sp}}$, and all $\ket{xy}, \ket{x'y'} \in \Hp_n$ with $\ket{xy} \nsim \ket{x'y'}$.
\end{proposition}

\begin{proof} 	
	In the following, the generic symbol $c$ denotes an inessential quantity that only depends on $d,s,p$ and can change from line to line. Quantities with a tilde refer to the SQW with the identity as scattering matrices. Since $\Hp^{uv}$ is invariant under $\tilde{U}_\omega^{B_n}$, we have
    \begin{align} \label{eq:res local on Hjk}
        \langle uv \, | \, \tilde{R}_\omega^{B_n} \, u'v' \rangle = 0 \;\; \text{ if } \; \ket{uv} \nsim \ket{u'v'}.
    \end{align}
    The resolvent identity implies
    \begin{align} \label{eq:res id C - C0}
        R_\omega^{B_n} = \tilde{R}_\omega^{B_n} + R_\omega^{B_n} \Big( \tilde{U}_\omega^{B_n} - U_\omega^{B_n} \Big) \tilde{R}_\omega^{B_n}.
    \end{align}
    By definition of the scattering matrices, we have:
    \begin{align} \label{eq:U nearest nbh}
        \langle uv \, | \, U_\omega^{B_n} \, u'v' \rangle = 0 \;\; \text{ if } \; v \neq u'.
    \end{align}
    Equations \eqref{eq:res local on Hjk} - \eqref{eq:U nearest nbh} now yield for $\ket{xy}, \ket{x'y'} \in \Hp_n$ with $\ket{xy} \nsim \ket{x'y'}$ and $z \in \C \setminus \Sp^1$:
    \begin{align*}
        \big| \langle xy \, | \, R_\omega^{B_n} \, x'y' \rangle \big| &= \big| \langle xy \, | \, R_\omega^{B_n} \Big( \tilde{U}_\omega^{B_n} - U_\omega^{B_n} \Big) \tilde{R}_\omega^{B_n} \, x'y' \rangle \big| \\ 
        &\leq \sum_{\substack{\ket{uv}, \ket{vv'} \in \Hp_n \\ \text{s.t } \ket{vv'} \sim \ket{x'y'}}} \big| \langle xy \, | \, R_\omega^{B_n} \, uv \rangle \, \langle uv \, | \, \Big( \tilde{U}_\omega^{B_n} - U_\omega^{B_n} \Big) \, vv' \rangle \, \langle vv' \, | \, \tilde{R}_\omega^{B_n} \, x'y' \rangle \big|.
    \end{align*}
    Using \eqref{eq: scat mat norm} we can bound the second scalar product by $\| \mathcal{S} - \mathcal{I} \|$. Since $U_\omega^{B_n}$ is a normal operator, we can bound the third scalar product by $\frac{1}{\text{dist} \big( z, \sigma(\tilde{U}_\omega^{B_n}) \big) }$. Finally, we note that the number of terms in the sum above is finite and independent of $n$. Thus, there exists a constant $c > 0$ such that
    \begin{align} \label{eq:resolvent bound}
        \big| \langle xy \, | \, R_\omega^{B_n} \, x'y' \rangle \big| &\leq c \, \frac{\| \mathcal{S} - \mathcal{I} \|}{\text{dist} \big( z, \sigma(\tilde{U}_\omega^{B_n}) \big)} \, \sup_{\substack{\ket{uv} \in \Hp_n \text{ s.t} \\ v=x' \text{ or } v = y'}} \big| \langle xy \, | \, R_\omega^{B_n} \, uv \rangle \big|.
    \end{align}
    For any $\eta > 0$, we define the set of "good" events:
    \begin{align*}
        G_\eta(z) = \left\{ \omega \in \T^{V} \text{ s.t. } \text{dist}\Big(z, \sigma \big( \tilde{U}_\omega^{B_n} \big) \Big) > \eta \right\}.
    \end{align*}
    By Proposition \ref{propo: spect gaps fully loc}, we have $\Pp \left( G_\eta(z)^c \right) \leq c \, \eta \, |B_n|$ for all $0<\eta<1$. We take $p > 1/(1-s)$ as in the statement of Proposition \ref{propo: poly decay resolvent} and $1<q<1/s$ such that $\frac{1}{p} + \frac{1}{q} = 1$ and apply Hölder's inequality and Theorem \ref{thm: frac mom bound}:
    \begin{align} \label{eq: first part of expectation}
    \begin{split}
        \E &\left( \chi_{G_\eta(z)^c} \, \left| \langle xy \, | \, R^{B_n} \, x'y' \rangle \right|^s \right)
        \leq \Pp\left(G_\eta(z)^c \right)^\frac{1}{p} \, \E \left( \left| \langle xy \, | \, R^{B_n} \, x'y' \rangle \right|^{sq} \right)^\frac{1}{q} \leq c \big(\eta \, |B_n| \big)^\frac{1}{p}.
    \end{split}
    \end{align}
    Note that Theorem \ref{thm: frac mom bound} requires $sq < 1$, which is equivalent to $p > \frac{1}{1-s}$. 
   By perturbation theory for normal operators, if  $\| \mathcal{S} - \mathcal{I} \| \leq \frac{\eta}{2}$, then 
    \begin{align} \label{eq:spectral implication}
        \text{dist} \Big(z, \sigma \big( \tilde{U}_\omega^{B_n} \big) \Big) > \eta \Longrightarrow \text{dist} \Big(z, \sigma \big( U_\omega^{B_n} \big) \Big) > \frac{\eta}{2}.
    \end{align}
    If $\omega \in G_\eta(z)$, we have $\text{dist} \big( z, \sigma(\tilde{U}_\omega^{B_n}) \big) > \eta$ and therefore by \eqref{eq:spectral implication} and \eqref{eq:resolvent bound}:
    \begin{align} \label{eq: almost final bound}
        \chi_{G_\eta(z)}(\omega) \, \left| \langle xy \, | \, R_\omega^{B_n} \, x'y' \rangle \right|^s \leq \chi_{G_\eta(z)}(\omega) \, (2c)^s \, \frac{\| \mathcal{S} - \mathcal{I} \|^s}{\eta^{2s}} .
    \end{align}
    The estimates \eqref{eq: first part of expectation} and \eqref{eq: almost final bound} thus yield
    \begin{align}
    	 \E &\left(  \left| \langle xy \, | \, R_\omega^{B_n} \, x'y' \rangle \right|^s \right)\leq  c\Big( \big(\eta \, |B_n| \big)^\frac{1}{p} + \frac{\| \mathcal{S} - \mathcal{I} \|^s}{\eta^{2s}} \Big).
    \end{align}
    At this point we balance the two contributions to get an expression for $\eta$
    \begin{align}
    	\big(\eta \, |B_n| \big)^\frac{1}{p}=\bigg( \frac{\| \mathcal{S} - \mathcal{I} \|}{\eta^2} \bigg)^s \ \Rightarrow \
    	\eta = \bigg(\frac{\| \mathcal{S} - \mathcal{I} \|^{sp}}{|B_n|}\bigg)^{\frac{1}{1+2sp}}.
    \end{align} 
    In case
    $\|\mathcal S-\mathcal I\|\leq \eta/2$, this yields in turn the sought for estimate
    \begin{align}
    	 \E \Big( \left| \langle xy \, | \, R_\omega^{B_n}(z) \, x'y' \rangle \right|^s \Big) &\leq c\left(\|\mathcal S-\mathcal I\||B_n|^2\right)^{\frac{s}{1+2sp}}.
    \end{align}
    Using $|B_n|\geq 1$, $d\geq 2$, one finally checks that the assumption
     \begin{align}\label{eq: estim S-I prop3}
     	\| \mathcal{S} - \mathcal{I} \| |B_n|^2 d^{\frac{1+2ps}{sp}}<1
     	\end{align}
     	ensures that the previous condition holds, which ends the proof.  
\end{proof} \noindent

\subsection{The iterative step} \label{sec: iterative step}
The goal of this section is the following
\begin{proposition} \label{propo: max on boundary}
    Under the assumptions of Theorem \ref{thm: dyn loc SQW}, we have for all $ 0<q<1$,  all $s \in \left(0, \frac{1}{3} \right)$ that there exist $N_0\in \N^*$ and $\varphi:\N^*\rightarrow \R^*_+$, (given in \eqref{eq: def phin} below), both depending on $(q,s,d)$, such that $n\geq N_0$ and $\| \mathcal{S} - \mathcal{I} \| \leq \varphi(n)$ imply
    \begin{align}\label{eq: est prop 4} 
        \E \left( \left| \langle xy \, | \, R(z) \, x'y' \rangle \right|^s \right) \leq q \, 
        \max_{\substack{\ket{uv} \in \Hp \\ \dist(v,x)\in \{n, n+1,n+2\}} }
        	\E \left( \left| \langle uv \, | \, R(z) \, x'y' \rangle \right|^s \right)
    \end{align}
    for all $z \in \C$ with $|z|\in [1/2, 3/2]\setminus \{1\}$,  $\ket{xy} \in \Hp$ and $\ket{x'y'} \in \Hp_{n+3}^\perp$, where the ball underlying the definition of $\Hp_{n+3}$ is centered at $x$.
\end{proposition}
\begin{proof}
    We use the resolvent identity twice, once to decouple over the boundary of $B_n \equiv B_n(x) $ and then over that of $ B_{n+1} \equiv B_{n+1}(x) $. Omitting the dependencies on $z$ and $\omega$ in the respective resolvents, 
  this yields
    \begin{align*}
        \langle xy \, | \, R \, x'y' \rangle 
        &= \langle xy \, | \, R^{(n)} \, x'y' \rangle - \langle xy \, | \, R^{(n)} T^{(n)} R^{(n+1)} \, x'y' \rangle + \langle xy \, | \, R^{(n)} T^{(n)} R T^{(n+1)} R^{(n+1)} \, x'y' \rangle.
    \end{align*}
    The first of the three terms is zero, since $\ket{xy} \in \Hp_n$ and $\ket{x'y'} \in \Hp_n^\perp$ and both subspaces are invariant under $R^{(n)}$. Similarly, since $\ket{x'y'} \in \Hp_{n+3}^\perp$ and using that $T^{(n)}$ has range $1$, we conclude that the second scalar product is also zero. Inserting the identity operator between each of the remaining terms, we conclude:
    \begin{align} \label{eq: much before first resampling}
        \nonumber \langle xy \, | \, R \, x'y' \rangle 
        = \sum_{\substack{\ket{uv}, \ket{u'v'},\\ \ket{wt}, \ket{w't'} \in \Hp}} \langle xy \, | \, R^{(n)} \, uv \rangle \, \langle uv \, | \, T^{(n)} \, u'v' \rangle \, \langle u'v' \, | \, R \, wt \rangle \, \times \quad & \\[-3ex] \langle wt \, | \, T^{(n+1)} \, w't' \rangle \, \langle w't' \, | \, R^{(n+1)} \, x'y' \rangle & .
    \end{align}
    We now introduce the abbreviation
\begin{align} \label{eq: def boundary}
    \partial B_n \coloneqq \big\{ \big(\ket{uv}, \ket{u'v'} \big) \in \Hp \times \Hp \, \text{ s.t. } \langle uv | \, T^{(n)} \, u'v' \rangle \neq 0 \big\} , 
\end{align}
and stress that this set does not coincide with boundary $V_{\partial F}$ for a consistent set $F$ defined in Section \ref{sec:RSQW}, since $\partial B_n$ is a subset of directed edges, while $V_{\partial F}$ is a subset of vertices. Since $x\in V$ is fixed in this proof,  for any $v\in V$,  we abbreviate by $\dist(v)$ the graph distance $d(v,x)$. Note that if $\big(\ket{uv}, \ket{u'v'} \big) \in \partial B_n$, then $u'=v$ and $\dist(v) = n$.
    We now bound the moduli of terms in~\eqref{eq: much before first resampling} containing the transition operators $T^{(n)}$ and $T^{(n+1)}$ by 2,
    and obtain, using $0<s<1$,
    \begin{align} \label{eq: before first resampling}
    \begin{split}
        &\E \big( \big| \langle xy \, | \, R \, x'y' \rangle \big|^s \big) \\ &\leq c \, \sum_{\substack{(\ket{uv}, \ket{u'v'}) \in \partial B_n, \\ (\ket{wt}, \ket{w't'}) \in \partial B_{n+1}, \\ \ket{uv} \in \Hp_n, \ket{w't'} \in \Hp_{n+1}^\perp}} \E \left( \left| \langle xy \, | \, R^{(n)} \, uv \rangle \right|^s \left| \langle u'v' \, | \, R \, wt \rangle \right|^s \left| \langle w't' \, | \, R^{(n+1)} \, x'y' \rangle \right|^s \right).
    \end{split}
    \end{align}
    Subsequently, $c$ denotes a constant independent of $n$ that may change from line to line.  Note that even though the first and third terms are independent, they are coupled by the second term. To decouple the terms, we use a resampling argument, which is detailed in Section~\ref{sec: first resampling argu}, and yields:
    \begin{equation} \label{eq: first resampling}
        \E \big( \big| \langle xy \, | \, R \, x'y' \rangle \big|^s \big) \leq c \mkern-5mu\sum_{\substack{\ket{wt} \in \Hp_n, \\ \dist(t) \in\{ n-1, n\}}} \E \left( \left| \langle xy \, | \, R^{(n)} wt \rangle \right|^s \right) \mkern-5mu\sum_{\substack{\ket{w't'}  \in \Hp_{n+1}^\perp, \\ \dist(t') \in \{ n+1, n+2\}}} \mkern-5mu \E \left( \left| \langle w't' \, | \, R^{(n+1)} \, x'y' \rangle \right|^s \right).
    \end{equation}
    Using the resolvent identity $R^{(n+1)} = R + R^{(n+1)} \, T^{(n+1)} \, R$ on the second term in \eqref{eq: first resampling} and inserting identities between each of the operators, we obtain for fixed $\ket{w't'}$:
    \begin{align*}
        \E &\Big(\left| \langle w't' \, | \, R^{(n+1)} \, x'y' \rangle \right|^s \Big) \leq \E \left(\left| \langle w't' \, | \, R \, x'y' \rangle \right|^s \right) \\
        &+ \E \bigg( \sum_{\ket{uv}, \ket{u'v'} \in \Hp} \left| \langle w't' \, | \, R^{(n+1)} uv \rangle \right|^s \, \left| \langle uv \, | \, T^{(n+1)} u'v' \rangle \right|^s \, \left| \langle u'v' \, | \, R \, x'y' \rangle \right|^s \bigg).
    \end{align*}
    Since $\ket{w't'} \in \Hp_{n+1}^\perp$, we need $(\ket{uv}, \ket{u'v'}) \in \partial B_{n+1}$ with $\ket{uv} \in \Hp_{n+1}^\perp$ to get a nonzero contribution. Bounding $T^{(n+1)}$, we have
    \begin{align*}
        \E \Big(\big| \langle w't' \, | \, &R^{(n+1)} \, x'y' \rangle \big|^s \Big) \leq \E \left(\left| \langle w't' \, | \, R \, x'y' \rangle \right|^s \right) \\ 
        &+ c \sum_{\substack{(\ket{uv}, \ket{u'v'}) \in \partial B_{n+1}, \\ \ket{uv} \in \Hp_{n+1}^\perp}} \E \bigg( \left| \langle w't' \, | \, R^{(n+1)} \, uv \rangle \right|^s \, \left| \langle u'v' \, | \, R \, x'y' \rangle \right|^s \bigg).
    \end{align*}
Plugging this into \eqref{eq: first resampling} yields:
    \begin{align} \label{eq: before second resampling}
        \E \big( \big| \langle xy \, | \, R \, x'y' \rangle \big|^s \big) \leq &c \, \bigg( \sum_{\substack{\ket{wt} \in \Hp_n, \\ \dist(t) \in \{ n-1, n\}}} \E \left( \left| \langle xy \, | \, R^{(n)} \, wt \rangle \right|^s \right) \bigg) 
        \sum_{\substack{\ket{w't'}  \in \Hp_{n+1}^\perp, \\ \dist(t') \in \{n+1, n+2\}}} \Bigg( \E \left(\left| \langle w't' \, | \, R \, x'y' \rangle \right|^s \right) \nonumber \\ &+ \tilde{c} \sum_{\substack{(\ket{uv}, \ket{u'v'}) \in \partial B_{n+1}, \\ \ket{uv} \in \Hp_{n+1}^\perp}} \E \bigg( \left|\langle w't' \, | \, R^{(n+1)} \, uv \rangle \right|^s \, \left| \langle u'v' \, | \, R \, x'y' \rangle \right|^s \bigg) \Bigg).
    \end{align}
    A second re-sampling argument explained in the appendix allows to estimate \eqref{eq: before second resampling}. We express the result in terms of $s_n=c e^{\gamma n}$, a common upper bound on $|S_n|$ and $|B_n|$ (using the bound $|B_n| \leq \sum_{k=0}^n |S_n| \leq c e^{\gamma n}$) :
    \begin{align} \label{eq: second resampling}
    \begin{split}
        \E \big( &\big| \langle xy \, | \, R \, x'y' \rangle \big|^s \big) \\ &\leq c \, s_n \, \sum_{\substack{\ket{wt} \in \Hp_n, \\ \dist(t) \in \{ n-1, n\} }} \E \left( \left| \langle xy \, | \, R^{(n)} \, wt \rangle \right|^s \right) \, \sum_{\substack{\ket{uv} \in \Hp, \\ \dist(v) \in \{ n, n+1, n+2\} }} \E \left( \big| \langle uv \, | \, R \, x'y' \rangle \big|^s \right).
    \end{split}
    \end{align}
    We prove \eqref{eq: second resampling} in Section \ref{sec: second resampling argu}. Note that since $\ket{xy}, \ket{wt} \in \Hp_n$, we can replace $R^{(n)}$ by $R^{B_n}$. Thus, assuming $|B_n|^2\| \mathcal{S} - \mathcal{I} \| < d^{-\frac{1+2ps}{sp}}$ for some $p > \frac{1}{1-s}$,
   we recall \eqref{eq: def sn} and apply Proposition~\ref{propo: poly decay resolvent} to the first sum in \eqref{eq: second resampling} to get 
    \begin{align*}
        \sum_{\substack{\ket{wt} \in \Hp_n, \\ \dist(t) \in \{ n-1, n\} }} \E \left( \left| \langle xy \, | \, R^{(n)} \, wt \rangle \right|^s \right) \leq c \, s_n \left(\|\mathcal S-\mathcal I\||B_n|^2\right)^\frac{s}{1+2sp},
    \end{align*}
    We have used, and will use again, that for $p\in \Z$ fixed, $e^{-\gamma |p|}s_n\leq s_{n+p}\leq e^{\gamma |p|}s_n.$
    The second sum is bounded according to
    \begin{align*}
        \sum_{\substack{\ket{uv} \in \Hp, \\ \dist(v) \in \{n, n+1, n+2\}}} \E \left( \big| \langle uv \, | \, R \, x'y' \rangle \big|^s \right) \leq c \, s_n \, \max_{\substack{\ket{uv} \in \Hp, \\ \dist(v) \in \{n, n+1, n+2\}}} \E \left( \big| \langle uv \, | \, R \, x'y' \rangle \big|^s \right).
    \end{align*}
    Inserting the last two estimates into \eqref{eq: second resampling}, we obtain
    \begin{align*}
        \E \big( &\big| \langle xy \, | \, R \, x'y' \rangle \big|^s \big) \leq c \, s_n^3 \left(\|\mathcal S-\mathcal I\||B_n|^2\right)^\frac{s}{1+2sp} \, \max_{\substack{\ket{uv} \in \Hp, \\ \dist(v) \in \{n, n+1, n+2\}}} \E \left( \big| \langle uv \, | \, R \, x'y' \rangle \big|^s \right).
    \end{align*}
    Using \eqref{eq: def sn} and $|B_n| \leq c e^{\gamma n}$, we impose that the prefactor satisfies 
    \begin{align}
    	c \, s_n^3 \left(\|\mathcal S-\mathcal I\||B_n|^2\right)^\frac{s}{1+2sp} \leq \tilde{c} \, e^{\gamma(3+\frac{2s}{1+2sp})n}\|\mathcal S-\mathcal I\|^\frac{s}{1+2sp}\leq q,
    \end{align}
    where $0<q<1$ and $\tilde{c}$ is a constant depending on $q,s,p$. Choosing $p(s)>1/(1-s)$, this condition is satisfied if
    \begin{align}\label{eq: def phin}
    	 \|\mathcal S-\mathcal I\|\leq e^{-\gamma \big(\frac{3+6sp(s)+2s}{s}\big)n }\big(q/\tilde{c}\big)^{\frac{1+2sp(s)}{s}}:=\varphi(n). 
    \end{align}
    To ensure \eqref{eq: estim S-I prop3}, we take $n\geq N_0$,  which ends the proof of Proposition \ref{propo: max on boundary}.
\end{proof}

\subsection{Exponential decay} \label{sec: main proof}
Repeated applications of Proposition \ref{propo: max on boundary} allow us to prove Theorem \ref{thm: expo dec frac mom}:

\begin{proof}
    We take $N_0\geq 2$, $q$ from Proposition \ref{propo: max on boundary} and set $n = N_0$. The LHS of \eqref{eq: expo dec of frac mom} is bounded due to Theorem \ref{thm: frac mom bound}, which is why it suffices to consider only $\ket{xy}$ and $\ket{x'y'}$ such that $|x-x'| > n+3$. This condition guarantees that $\ket{x'y'} \in \Hp_{n+3}^\perp$, since $\Hp_{n+3}$ is centered at $x$.  An application of Proposition \ref{propo: max on boundary} yields for the fixed value $\| \mathcal{S} - I \| = \varphi(n)= \varphi(N_0)$: 
    \begin{align} \label{eq: finite volume met}
        \E \left( \left| \langle xy \, | \, (U - z)^{-1} \, x'y' \rangle \right|^s \right) \leq q \, \max_{\substack{\ket{uv} \in \Hp \\ \dist(v,x)\in \{n, n+1,n+2\}}} \E \left( \left| \langle uv \, | \, (U - z)^{-1} \, x'y' \rangle \right|^s \right).
    \end{align}
    We define 
    \begin{align}
    	r(x)\coloneqq\max_{y\sim x}  \E \left( \left| \langle xy \, | \, (U - z)^{-1} \, x'y' \rangle \right|^s \right) 
    \end{align}
    and use the fact that for $u\sim v$, $\dist(u,x)\in \{m-1,m,m+1\}$ if $\dist(v,x)=m$ to get from \eqref{eq: finite volume met}
    \begin{align}
    	r(x)\leq q \max_{\substack{u \in V \\ \dist(u,x)\in A_n(x)}} r(u), 
    \end{align}
    where $A_n(x)$ is the annulus
    \begin{align}
    	A_n(x)=\{v\in V \ | \ \dist(v,x)\in \{n-1, \dots, n+3\} \}, \ \text{ with } \ |A_n(x)|\leq c s_n.
    \end{align} 
    Let $u_1\in A_n(x)$ be the locus of the maximum above and consider $u_2\in A_n(u_1)$ such that 
    \begin{align}
    	r(x)\leq q r(u_1)\leq q^2 \max_{\substack{u \in V \\ \dist(u,u_1)\in A_n(x)}} r(u) = q^2 r(u_2).
    \end{align}
    The vertices $x, u_1, u_2$ are distinct since $q<1$ and $r(x)>0$, and $\dist(x,u_2)\leq 2(n+3)$. We can iterate this procedure 
    at least $M= \lfloor \frac{\dist(x,x')}{n+3} \rfloor$ times, so that $\dist(u_k,x')>n+3$, $k\in \{1, \dots, M\}$, and apply Theorem \ref{thm: frac mom bound} to $r(u_M)$ to obtain 
    \begin{align}
    	r(x)\leq q^M r(u_M) \leq q^M C.
    \end{align}
     Setting $g = \frac{|\log(q)|}{n+3}$ eventually yields 
    \begin{align*}
        \E \left( \left| \langle xy \, | \, (U - z)^{-1} \, x'y' \rangle \right|^s \right) \leq c \, e^{-g\dist(x,x')}.
    \end{align*}
    To finish the proof, we observe that as $q\rightarrow 0$, $n=N_0$ defined in \eqref{eq: def phin} can be taken as $N_0=1$, so that $g$ behaves as $|\ln q|$ in this limit. Hence, $g>0$ can be chosen arbitrarily large, taking $0<q<1$ arbitrarily small, as announced. 
\end{proof}
We are now ready to prove Theorem \ref{thm: dyn loc SQW} as an application of Theorem \ref{thm:fmec}:
\begin{proof}[Proof of Theorem \ref{thm: dyn loc SQW}]
    Given $ x \in V $ and $L \in \mathbb{N} $, we define for 
    the set of oriented edges
    \begin{align*}
        E_L = \{ (uv) \in D \, | \, d(x,u) \leq L, d(x,v) \leq L \} .
    \end{align*}
    This is an element of the consistent family $\mathcal{F} = \{ F \subset D \, | \, (uv) \in F \iff (vu) \in F \}$, see Section~\ref{sec:RSQW}. Since $E_L \nearrow D$, we have  
    $
        \lim_{L \to \infty} \| T_\omega^{D, E_L} \ket{uv}  \| = 0 $ for any fixed edge $\ket{uv} \in \Hp$. 
    By Proposition~\ref{propo: EC lower semicont}, we bound the eigenfunction correlator $Q_\omega^D$ on the set of all directed edges $D$ by its analogue $Q_\omega^{E_L}$ on the finite set $E_L$, see \eqref{eq:ECconv}. We choose $L$ big enough such that $(xy), (x'y') \in E_L$, and consider from now on the random SQW $U_\omega^{E_L}$. 
    
Since the LHS in~\eqref{eq:dynloc} is bounded by one, we may assume without loss of generality that $ xy $ and $ x'y' $, when considered as two vertex pairs, have a graph distance greater or equal to $ 2 r \geq 4 $ with some $ r < L /2 $. This $ r $ can be chosen such that $ d(x,x')\leq 2 (r+1) $ and that $(xy) \in E_r $ and $(x'y') \in E_r^c $. We then apply Theorem \ref{thm:fmec} with $B = E_r \subsetneq E_L $. The fractional moment and spectral averaging conditions \eqref{eq: assu frac mom} and \eqref{eq: assu spec avg} are satisfied due to Theorems \ref{thm: frac mom bound} and \ref{thm: spec avg}. The latter uses that $(xy)$ and $(x'y')$ are far enough apart such that $\omega_{x'}$ is independent of the randomness in $B$. Theorem \ref{thm:fmec} thus implies that for any $s \in (0, 1/3)$, $g>0$ and $\beta \in (0,s)$ with $\beta \leq 1 - \beta / s$ and all open $I \subset \Sp^1$
    \begin{align}
    \nonumber \mathbb{E}&\left[Q^{E_L}((x'y'),(xy);I,\beta)\right] \\ &\leq C_W^{\frac{\beta}{s}} \sum_{(uv)\in B, (u'v') \in E} t_{(uv),(u'v')}^\beta \sup_{\delta \in (0,1) } \left\{ \int_I \mathbb{E}\left[ |\langle xy | (U^B - \delta e^{i\theta})^{-1} uv \rangle|^{s} \right] \frac{d\theta}{2\pi}\right\}^{\frac{\beta}{s}} .
    \end{align}
    We note that $t_{(uv),(u'v')}=0$ unless $v=u'$ and $d(x,v) = r$, in which case it can be bounded by $2$. By Theorem \ref{thm: expo dec frac mom}, there exist $\varphi, c > 0$ such that the fractional moment on the RHS is bounded by $c \, e^{-g d(x,u)}$ whenever $\| \mathcal{S} - \mathcal{I} \| \leq \varphi$. Since $u \sim v$, we have $r-1 \leq d(x,u) \leq r+1$. In total, this yields
    \begin{align*}
        \mathbb{E}\left[Q^{E_L}((x'y'),(xy);I,\beta)\right] \leq 2 C_W^\frac{\beta}{s} \sum_{\substack{v \in E_L, d(x,v) = r \\ u \sim v, v' \sim v}} \big( c \, e^{-g(r-1)} \big)^\frac{\beta}{s} \leq 2 C_W^\frac{\beta}{s} \, c^\frac{\beta}{s} \, d^2 \, s_r \, e^{g \, \frac{\beta}{s}} \, e^{-g \frac{\beta}{s} \, r} . 
    \end{align*}
    Using $s_r \leq c_d e^{\gamma r}$, see \eqref{eq: def sn}, and $r \geq d(x,x')/2  -1 $, we obtain for a new constant $c$
    \begin{align*}
        \mathbb{E}\left[Q^{E_L}((x'y'),(xy);I,\beta)\right] \leq c \, \exp(\frac{-(g \frac{\beta}{s} - \gamma)}{2} \, d(x,x')) .
    \end{align*}
    Since $g>0$ is not fixed, we can arbitrarily increase the rate of decay above. Due to all constants being independent of $L$, we obtain the same bound for $Q_\omega^D$ by Proposition \ref{propo: EC lower semicont}. Dynamical localization finally follows from \eqref{eq: loc from ec decay}.
\end{proof}

\appendix
\section{Resampling arguments}
\subsection{First re-sampling argument} \label{sec: first resampling argu}
We derive equation \eqref{eq: first resampling} from \eqref{eq: before first resampling}, following the steps of \cite[Prop. 13.1]{HJS:09}. We let $c$ denote a constant, which is independent of $n$, but may change from line to line.
\subsubsection{Splitting terms}
We fix $(\ket{uv}, \ket{u'v'}) \in \partial B_n, (\ket{wt}, \ket{w't'}) \in \partial B_{n+1}$ such that $\ket{uv} \in \Hp_n$ and $\ket{w't'} \in \Hp_{n+1}^\perp$ and define $J = \{v, v', t, t' \}$. 

For later use, we note that $\ket{uv} \in \Hp_n$, so $\dist(u) \leq n$ and $\dist(v) \leq n$. Additionally, due to $\langle uv | \, T^{(n)} \, u'v' \rangle \neq 0$, we actually have $u' = v$, $\dist(v)=n$ and $n-1\leq \dist(v')\leq n+1$.  On the other hand, due to $\ket{w't'} \in \Hp_{n+1}^\perp$ and $\langle wt | \, T^{(n+1)} \, w't' \rangle \neq 0$, we conclude that $w'=t$, and $\dist(t) = n+1$. Thus, $\dist(t')=n+2$ and $n \leq \dist(w) \leq n+2$.

\medskip

For all sites $r \in J$ we choose i.i.d. random variables $(\widehat{\omega}_r)_{r\in J}$ that are uniformly distributed on $\T$ and independent from $(\omega_r)_{r \in V}$. With $ P_r^{\rm O} $ the projection onto the outgoing subspace at $ r $ (see Section~\ref{sec:SQW}), we define 
\begin{align*}
    D_{\omega, \widehat{\omega}} = \sum_{r \in J} e^{i \widehat{\omega}_r} \, P_r^{\rm O} + \sum_{r \notin J} e^{i \omega_r} \, P_r^{\rm O}.
\end{align*}
$D_{\omega, \widehat{\omega}}$ uses the phase $\widehat{\omega}_r$ on all $r \in J$ and $\omega_r$ on all other sites. We set
\begin{align*}
    U_{\omega, \widehat{\omega}}^{(n)} = D_{\omega, \widehat{\omega}} \, U^{(n)} \;\; \text{and} \;\; \widehat{R}^{(n)} = \big( U_{\omega, \widehat{\omega}}^{(n)} - z \big)^{-1}.
\end{align*}
To simplify the notation we define $\widehat{D} = \sum_{r \in J} \big( e^{i \omega_r} - e^{i \widehat{\omega}_r} \big) \, P_r^{\rm O}$. Applying the resolvent identity to $R^{(n)}$ and $R^{(n+1)}$ yields
\begin{align*}
    R^{(n)} &= \widehat{R}^{(n)} - \widehat{R}^{(n)} \widehat{D} \, U^{(n)} R^{(n)} \;\; \text{ and } \;\; R^{(n+1)} = \widehat{R}^{(n+1)} - R^{(n+1)} \widehat{D} \, U^{(n+1)} \widehat{R}^{(n+1)}.
\end{align*}
Since $s<1$, we use this to bound each term in the RHS of \eqref{eq: before first resampling}:
\begin{align} \label{eq: splitting first res}
\begin{split}
    &\E \Big( \big| \langle xy \, | \, R^{(n)} \, uv \rangle \big|^s \, \big| \langle u'v' \, | \, R \, wt \rangle \big|^s \, \big| \langle w't' \, | \, R^{(n+1)} \, x'y' \rangle \big|^s \Big) \\
    &\leq \widehat{\E} \E \bigg( \Big( \big| \langle xy \, | \, \widehat{R}^{(n)} \, uv \rangle \big|^s + \big| \langle xy \, | \, \widehat{R}^{(n)} \widehat{D} \, U^{(n)} R^{(n)} \, uv \rangle \big|^s \Big) \, \big| \langle u'v' \, | \, R \, wt \rangle \big|^s \\
    &\hphantom{{} \leq \widehat{\E} \E \bigg(} \Big( \big| \langle w't' \, | \, \widehat{R}^{(n+1)} \, x'y' \rangle \big|^s + \big| \langle w't' \, | \, R^{(n+1)} \widehat{D} \, U^{(n+1)} \widehat{R}^{(n+1)} \, x'y' \rangle \big|^s \Big) \bigg) \\
    & = A_1 + A_2 + A_3 + A_4,
\end{split}
\end{align}
where
\begin{align*}
    A_1 &= \widehat{\E}\E \bigg( \Big| \langle xy \, | \, \widehat{R}^{(n)} \, uv \rangle \Big|^s  \; \Big| \langle u'v' \, | \, R \, wt \rangle \Big|^s \;  \Big| \langle w't' \, | \, \widehat{R}^{(n+1)} \, x'y' \rangle \Big|^s  \bigg)\\
    A_4 &= \widehat{\E}\E \bigg( \Big| \langle xy \, | \, \widehat{R}^{(n)} \widehat{D} \, U^{(n)} R^{(n)} \, uv \rangle \Big|^s \; \Big| \langle u'v' \, | \, R \, wt \rangle \Big|^s \\ & \hphantom{{} = \widehat{\E}\E \bigg(} \Big| \langle w't' \, | \, R^{(n+1)} \widehat{D} \, U^{(n+1)} \widehat{R}^{(n+1)} \, x'y' \rangle \Big|^s \bigg).
\end{align*}
The intermediate terms $A_2$ and $A_3$ are defined analogously. Here, $\E$ denotes the expectation over $(\omega_r)_{r \in V}$ and $\widehat{\E}$ the one over $(\widehat{\omega}_r)_{r \in J}$. 

\subsubsection{Estimate for $A_1$} \label{sec: estforA1}
Let 
$\E \big( ... |J \big)$
 be the conditional expectation with respect to the $\sigma$-field generated by $(\omega_r)_{r \notin J}$. By the law of total expectation, we have for any random variable $X$:
\begin{align} \label{eq: total expectation}
    \E (X) = \E \big( \E (X | J^c ) \big), \ \text{ where } \ 	\E \big( X|J^c \big)= \int_{\T}\dots\int_{\T} X \prod_{r\in J}d\theta_j/2\pi.
\end{align}
This gives us:
\begin{align*}
    A_1 &= \widehat{\E} \bigg( \E \bigg( \big| \langle xy | \, \widehat{R}^{(n)} uv \rangle \big|^s  \; \E \Big( \big| \langle u'v' | \, R \, wt \rangle \big|^s \, \Big| \, J^c \Big) \, \big| \langle w't' | \, \widehat{R}^{(n+1)} x'y' \rangle \big|^s \bigg) \bigg).
\end{align*}
We have used that the operator $U_{\omega, \widehat{\omega}}^{(n)}$ depends only on the random variables $(\omega_r)_{r \notin J}$ and $(\widehat{\omega}_r)_{r \in J}$. Thus, when taking the expectation $\E$, which is with respect to the random variables $(\omega_r)_{r \in V}$, the resolvents $\widehat{R}^{(n)}$ and $\widehat{R}^{(n+1)}$ are measurable with respect to the $\sigma$-algebra generated by $(\omega_r)_{r \notin J}$ and can be pulled out of the conditional expectation. By Theorem \ref{thm: frac mom bound}, we can bound the conditional expectation:
\begin{align*}
    \E \Big( \big| \langle u'v' \, | \, R \, wt \rangle \big|^s \, \Big| \, J^c \Big) \leq C(s),
\end{align*}
since $v', t \in J$. We obtain:
\begin{align*}
    A_1 \leq c \; \widehat{\E} \bigg( \E \Big( \big| \langle xy \, | \, \widehat{R}^{(n)} \, uv \rangle \big|^s  \; \big| \langle w't' \, | \, \widehat{R}^{(n+1)} \, x'y' \rangle \big|^s \Big) \bigg).
\end{align*}
Notice that the inequality above does not contain $(\omega_r)_{r \in J}$. Since $(\omega_r)_{r \in J}$ and $(\widehat{\omega}_r)_{r \in J}$ are identically distributed, we can replace the latter by the former:
\begin{align*}
    A_1 \leq c \; \E \Big( \big| \langle xy \, | \, R^{(n)} uv \rangle \big|^s  \; \big| \langle w't' \, | \, R^{(n+1)} \, x'y' \rangle \big|^s \Big).
\end{align*}
We stress that $\ket{uv} \in \Hp_n$, $\ket{w't'} \in \Hp_{n+1}^\perp$ and $\Hp_n$, $\Hp_{n+1}^\perp$ are invariant under $R^{(n)}$ and $R^{(n+1)}$ respectively. Thus, the two scalar products above are independent random variables, and we obtain:
\begin{align} \label{eq: A1 estimate}
    A_1 \leq c \; \E \Big( \big| \langle xy \, | \, R^{(n)} uv \rangle \big|^s \Big) \; \E \Big( \big| \langle w't' \, | \, R^{(n+1)} \, x'y' \rangle \big|^s \Big).
\end{align}
This estimate is of the required form.

\subsubsection{Estimate for $A_4$} \label{subsec: A4}
We apply \eqref{eq: total expectation} to $A_4$ and use Hölder's inequality twice:
\begin{align}\label{eq: A4 terms}
\begin{split} 
    A_4 &\leq \widehat{\E} \bigg( \E \bigg( \E \Big( \big| \langle xy \, | \, \widehat{R}^{(n)} \widehat{D} \, U^{(n)} R^{(n)} \, uv \rangle \big|^{3s} \Big| J^c \Big)^\frac{1}{3} \, \E \Big( \big| \langle u'v' \, | \, R \, wt \rangle \big|^{3s} \Big| J^c \Big)^\frac{1}{3} \\
    & \hphantom{{} \leq \widehat{\E} \bigg( \E \bigg( } \E \Big( \big| \langle w't' \, | \, R^{(n+1)} \widehat{D} \, U^{(n+1)} \widehat{R}^{(n+1)} \, x'y' \rangle \big|^{3s} \Big| J^c \Big)^\frac{1}{3} \bigg) \bigg).
\end{split}
\end{align}
Due to $s < \frac{1}{3}$, we can bound the second expectation in \eqref{eq: A4 terms} by Theorem \ref{thm: frac mom bound}. We continue with the first term. Using the definition of $\widehat{D}$ and that the number of terms in that definition is bounded independently of $n$, we get:
\begin{align*}
    \big| &\langle xy \, | \, \widehat{R}^{(n)} \widehat{D} \, U^{(n)} R^{(n)} \, uv \rangle \big|^{3s} 
    {\leq} \sum_{r \in J} \sum_{p \sim r} \big| e^{i \omega_r} - e^{i \widehat{\omega}_r} \big|^{3s} \, \big| \langle xy \, | \, \widehat{R}^{(n)} pr \rangle \big|^{3s} \, \big| \langle pr \, | \, U^{(n)} R^{(n)} \, uv \rangle \big|^{3s}.
\end{align*}
Note that $\ket{xy} \in \Hp_n$, which is invariant under $\widehat{R}^{(n)}$. Thus, $\ket{pr}$ is an element of $\Hp_n$ as well and  we obtain
\begin{align*}
    \big| &\langle xy | \, \widehat{R}^{(n)} \widehat{D} \, U^{(n)} R^{(n)} uv \rangle \big|^{3s} \leq c \sum_{r \in J} \sum_{\substack{ p \sim r, \\ \ket{pr} \in \Hp_n}} \big| \langle xy \, | \, \widehat{R}^{(n)} pr \rangle |^{3s} \, | \langle pr \, | \, U^{(n)} R^{(n)}  uv \rangle \big|^{3s}.
\end{align*}
We remark that $\widehat{R}^{(n)}$ depends only on $(\widehat{\omega}_r)_{r \in J}$ and $(\omega_r)_{r \notin J}$. Since the expectation $\E$ is with respect to $(\omega_r)_{r \in V}$, the first scalar product is measurable with respect to the $\sigma$-algebra generated by $(\omega_r)_{r \notin J}$. Therefore:
\begin{align}
\begin{split} \label{eq: A4 first term intermediate step}
    \E &\Big( \big| \langle xy \, | \, \widehat{R}^{(n)} \widehat{D} \, U^{(n)} R^{(n)} \, uv \rangle \big|^{3s} \Big| J^c \Big) 
    \leq c \sum_{r \in J} \sum_{\substack{ p \sim r, \\ \ket{pr} \in \Hp_n}} \big| \langle xy \, | \, \widehat{R}^{(n)} pr \rangle \big|^{3s} \, \E \Big( \big| \langle pr \, | \, U^{(n)} R^{(n)} \, uv \rangle \big|^{3s} \, \Big| \, J^c \Big).
\end{split}
\end{align}
Since $D_\omega$ is unitary, we have:
\begin{align*}
    \big| \langle pr \, | \, U^{(n)} R^{(n)} \, uv \rangle \big| = \big| \langle e^{i \omega_r} \, pr \, | \, I + z \, R^{(n)} \, uv \rangle \big|.
\end{align*}
Using this, Theorem \ref{thm: frac mom bound} and $1/2<|z|<3/2$, we can bound the expectation in \eqref{eq: A4 first term intermediate step}, since $r$ and $v$ belong to $J$, and obtain for the first term in \eqref{eq: A4 terms}:
\begin{align*}
    \E &\Big( \big| \langle xy \, | \, \widehat{R}^{(n)} \widehat{D} \, U^{(n)} R^{(n)} \, uv \rangle \big|^{3s} \Big| J^c \Big) \leq c \sum_{r \in J} \sum_{\substack{ p \sim r, \\ \ket{pr} \in \Hp_n}} \big| \langle xy \, | \, \widehat{R}^{(n)} \, pr \rangle \big|^{3s}.
\end{align*}

We can obtain a similar estimate for the third term in \eqref{eq: A4 terms}: As above, we use the definition of $\widehat{D}$, apply the expectation, use that $\widehat{R}^{(n+1)}$ is measurable with respect to $(\omega_r)_{r \notin J}$ and apply Theorem \ref{thm: frac mom bound} to obtain for $t'$ and $r$ in $J$:
\begin{align*}
    \E \Big( \big| \langle w't' | R^{(n+1)} \widehat{D} U^{(n+1)} \widehat{R}^{(n+1)} x'y' \rangle \big|^{3s} \Big| J^c \Big) \leq c \sum_{r \in J} \sum_{\substack{p \sim r, \\ \ket{pr} \in \Hp_{n+1}^c}} \big| \langle pr | U^{(n+1)} \widehat{R}^{(n+1)} x'y' \rangle \big|^{3s}.
\end{align*}
Due to our assumptions, $\ket{x'y'} \in \Hp_{n+3}^\perp$ so that $(pr) \neq (x'y')$. Indeed, by definition $r \in J = \{v, v', t, t'\}$ and $\ket{pr} \in \Hp_{n+1}^\perp$.  We conclude from the discussion  at the top of Section \ref{sec: first resampling argu} that $r \neq v$, $r=t'$ implies $\dist(r)=n+2$ and $r\in \{v',t\}$ implies $\dist(r)=n+1$ and $\dist(p)=n+2$.
This  yields $\dist(r) \in \{n+1, n+2\}$ while $\ket{x'y'} \in \Hp_{n+3}^\perp$ implies $\dist(y') \geq n+3$ and thus $(pr) \neq (x'y')$.

\medskip

Since $D_{\omega, \widehat{\omega}}$ is unitary, it follows:
\begin{align*}
    \big| \langle pr | U^{(n+1)} \widehat{R}^{(n+1)} x'y' \rangle \big| &= \big| \langle pr | I+ z \, \widehat{R}^{(n+1)} x'y' \rangle \big| = |z| \, \big| \langle pr | \widehat{R}^{(n+1)} x'y' \rangle \big|.
\end{align*}
Using $1/2<|z| < 3/2$, this yields for the third term in \eqref{eq: A4 terms}:
\begin{align*}
    \E \Big( \big| \langle w't' \, | \, R^{(n+1)} \widehat{D} U^{(n+1)} \widehat{R}^{(n+1)} x'y' \rangle \big|^{3s} \Big| J^c \Big) \leq c \sum_{r \in J} \sum_{\substack{p \sim r, \\ \ket{pr} \in \Hp_{n+1}^\perp}} \big| \langle pr | \widehat{R}^{(n+1)} x'y' \rangle \big|^{3s}.
\end{align*}
Since $J$ has a fixed finite number of elements and $3s<1$, we have $\big( \sum_{j \in J} \alpha_j^{3s} \big)^\frac{1}{3} \leq c \sum_{j \in J} \alpha_j^s$ for some constant $c$ and all $\alpha_j > 0$. Using the obtained bounds for the three terms in \eqref{eq: A4 terms}, we conclude:
\begin{align*}
    A_4 \leq c \, \sum_{r \in J} \sum_{\substack{ p \sim r, \\ \ket{pr} \in \Hp_n}} \, \sum_{r' \in J} \sum_{\substack{p' \sim r', \\ \ket{p'r'} \in \Hp_{n+1}^\perp}} \widehat{\E} \E \Big( \big| \langle xy \, | \, \widehat{R}^{(n)} pr \rangle \big|^{s} \, \big| \langle p'r' \, | \, \widehat{R}^{(n+1)} \, x'y' \rangle \big|^{s} \Big).
\end{align*}
We note that the last line only depends on the random variables $(\omega_r)_{r \notin J}$ and $(\widehat{\omega}_r)_{r \in J}$. Since $(\omega_r)_{r \in J}$ and $(\widehat{\omega}_r)_{r \in J}$ are independent and identically distributed, we can replace the latter by the former. The two scalar products inside the expectation are then independent random variables by the same reasoning as for $A_1$. This yields:
\begin{align} \label{eq: A4 estimate}
    A_4 \leq c \, \sum_{\substack{r \in J, \\ \ket{pr} \in \Hp_n}} \E \Big( \big| \langle | xy \, | \, R^{(n)} pr \rangle \big|^{s} \Big) \, \sum_{\substack{r' \in J, \\ \ket{p'r'} \in \Hp_{n+1}^\perp}}  \E \Big( \big| \langle p'r' \, | \, R^{(n+1)} \, x'y' \rangle \big|^{s} \Big).
\end{align}
We can interpolate the arguments used for $A_1$ and $A_4$ to obtain similar estimates for $A_2$ and $A_3$. We note that in the sum above, we have $\dist(r) \in \{n-1, n\}$ and $\dist(r') \in \{n+1, n+2 \}$. Using \eqref{eq: before first resampling} and \eqref{eq: splitting first res}, we finally obtain \eqref{eq: first resampling}. We stress that all constants obtained in this section are independent of $n$ and of the center of the ball $B_n$.

\subsection{Second re-sampling argument} \label{sec: second resampling argu}
We derive equation \eqref{eq: second resampling} from \eqref{eq: before second resampling}. Since the methods used are similar to the ones in Appendix~\ref{sec: first resampling argu}, we will omit details. We let $c > 0$ denote some constant independent of $n$ that may change from line to line. 

Fixing $\ket{w't'} \in \Hp_{n+1}^\perp$ with $\dist(t') \in \{ n+1, n+2\}$, $(\ket{uv}, \ket{u'v'}) \in \partial B_{n+1}$ with $\ket{uv} \in \Hp_{n+1}^\perp$, we define $\Tilde{J} = \{ t', v, v' \}$. Note that $\dist(v)=n+1$ and $n\leq \dist(v)\leq n+2$. Similar to our approach in Appendix \ref{sec: first resampling argu}, we re-sample the random variables on all sites in $\tilde{J}$. Let $(\tilde{\omega}_r)_{r \in \tilde{J}}$ be an i.i.d. family of random variables uniformly distributed on the torus $\T$. Let $D_{\omega, \tilde{\omega}}$ use the phase $\tilde{\omega}_r$ on all sites $r \in \tilde{J}$ and $\omega_r$ on all other sites. Let $U_{\omega, \tilde{\omega}} = D_{\omega, \tilde{\omega}} U$ and $\tilde{R}$ be the re-sampled quantum walk and its resolvent. To simplify the notation we define $\tilde{D} = \sum_{r \in \tilde{J}} \big( e^{i \omega_r} - e^{i \tilde{\omega}_r} \big) \, P_r^{\rm O}$. Using the resolvent identity $R = \tilde{R} - R \, \tilde{D} \, U \, \tilde{R}$, we obtain for the last expectation in \eqref{eq: before second resampling}:
\begin{align} \label{eq: split in B1 and B2}
\begin{split}
    \E \bigg( \left|\langle w't' \, | \, R^{(n+1)} \, uv \rangle \right|^s \, \left| \langle u'v' \, | \, R \, x'y' \rangle \right|^s \bigg) 
    &{\leq} \tilde{\E} \E \bigg( \left|\langle w't' \, | \, R^{(n+1)} \, uv \rangle \right|^s \, \big| \langle u'v' \, | \, \tilde{R} \, x'y' \rangle \big|^s \bigg) \\ + \tilde{\E} \E \bigg( \left|\langle w't' \, | \, R^{(n+1)} \, uv \rangle \right|^s \, &\big| \langle u'v' \, | \, R \tilde{D} U \tilde{R} \, x'y' \rangle \big|^s \bigg) = B_1 + B_2.
\end{split}
\end{align}
We write $\E (... | \tilde{J})$ for the conditional expectation with respect to the $\sigma$-algebra generated by $(\omega_r)_{r \notin \tilde{J}}$. We proceed as in section \ref{sec: estforA1} and apply the law of total expectation (equation \eqref{eq: total expectation}) to $B_1$, use that $\tilde{R}$ is independent of the random variables $(\omega_r)_{r \in \tilde{J}}$ and apply Theorem \ref{thm: frac mom bound}, since $\{t',v\}\subset\Tilde{J}$. We can then replace the random variables $(\tilde{\omega}_r)_{r \in \tilde{J}}$ by $(\omega_r)_{r \in \tilde{J}}$ to obtain:
\begin{align} \label{eq: B1 estimate}
    B_1 \leq c \, \E \Big( \big| \langle u'v' \, | \, R \, x'y' \rangle \big|^s \Big).
\end{align}
The law of total expectation, Hölder's inequality and Theorem \ref{thm: frac mom bound} give for $B_2$:
\begin{align} \label{eq: B2 first estimate}
    B_2 \leq c \, \tilde{\E} \E \Big( \E \Big( \big| \langle u'v' \, | \, R \tilde{D} U \tilde{R} \, x'y' \rangle \big|^{2s} \, \Big| \, \tilde{J}^c \Big)^\frac{1}{2} \Big).
\end{align}
Using the definition of $\tilde{D}$, $2s<1$ and Theorem \ref{thm: frac mom bound}, we obtain for the inner expectation:
\begin{align} \label{eq: B2 second estimate}
    \E &\Big( \big| \langle u'v' \, | \, R \tilde{D} U \tilde{R} \, x'y' \rangle \big|^{2s} \, \Big| \, \tilde{J}^c \Big) \leq c \sum_{\substack{\ket{pr} \in \Hp \\ \text{s.t. } r \in \tilde{J}}} \big| \langle pr \, | \, U \tilde{R} \, x'y' \rangle \big|^{2s}.
\end{align}
By the same argument as in Appendix \ref{sec: first resampling argu}, we have for every $r \in \tilde{J}$ that $|r| \leq n+2$. Since $\ket{x'y'} \in \Hp_{n+3}^\perp$, it follows that $y' \notin \tilde{J}$. This yields:
\begin{align} \label{eq: B2 third estimate}
    \big| \langle pr | \, U \tilde{R} \, x'y' \rangle \big| = \big| \langle pr | \, \big( I + z \tilde{R} \big) \, x'y' \rangle \big| = |z| \, | \langle pr | \, \tilde{R} \, x'y' \rangle |.
\end{align}
Since $|\tilde{J}| \leq 3$, we have $\big( \sum_{j \in \tilde{J}} \alpha_j^{2s} \big)^\frac{1}{2} \leq c \sum_{j \in \tilde{J}} \alpha_j^{s}$ for all $\alpha_j > 0$. Using $|z|<3/2$, we insert \eqref{eq: B2 second estimate} and \eqref{eq: B2 third estimate} into \eqref{eq: B2 first estimate} and replace $\tilde{\omega}_r$ by $\omega_r$ to obtain:
\begin{align} \label{eq: B2 estimate}
    B_2 \leq c \sum_{\substack{\ket{pr} \in \Hp \\ \text{s.t. } r \in \tilde{J}}} \E \Big( | \langle pr \, | \, R \, x'y' \rangle |^{s} \Big).
\end{align}
Plugging the estimates \eqref{eq: split in B1 and B2}, \eqref{eq: B1 estimate} and \eqref{eq: B2 estimate} into \eqref{eq: before second resampling} we obtain:
\begin{align*}
    \E \big( \big| &\langle xy | \, R \, x'y' \rangle \big|^s \big) \leq c \, \bigg( \sum_{\ket{wt} \in \Hp_n, \dist(t) \in\{ n-1, n\}} \E \left( \big| \langle xy | \, R^{(n)} \, wt \rangle \big|^s \right) \bigg) \\ &\times \sum_{\substack{ \ket{w't'} \in \Hp_{n+1}^\perp \\ \dist(t') \in \{n+1, n+2\}}} \Bigg( \E \left( \big| \langle w't' | \, R \, x'y' \rangle \big|^s \right) + c \sum_{\substack{(\ket{uv}, \ket{u'v'}) \in \partial B_{n+1}, \\ \ket{uv} \in \Hp_{n+1}^\perp}} \Bigg(\E \Big( \big| \langle u'v' | \, R \, x'y' \rangle \big|^s \Big) \\ &\hphantom{{} \times \sum_{\substack{ \ket{w't'} \in \Hp_{n+1}^\perp \\ \dist(t') \in \{n+1, n+2\}}} \Bigg(} + \sum_{\substack{\ket{pr} \in \Hp \\ r \in \{ t', v, v'\} }} \E \Big( | \langle pr | \, R \, x'y' \rangle |^{s} \Big) \Bigg) \Bigg).
\end{align*}
 The sum over $\ket{w't'}$ leads to at most $d \, (|S_{n+1}| + |S_{n+2}|) \leq  c s_n$ terms. Similarly, the sum over $\ket{uv}, \ket{u'v'}$ leads to $\mathcal{O}(s_n)$ terms. We can therefore replace the last three expectations by a single sum of $\E \big( | \langle uv \, | R \, xy \rangle |^s \big)$ over $\ket{uv} \in \Hp$ with $\dist(v) \in \{n, n+1, n+2\}$. This leads to an additional term of the order $\mathcal{O}(s_n)$ in \eqref{eq: second resampling}.

\paragraph{Acknowledgements.} 	
AJ and AS were partially supported by the French National Research
Agency, ANR Dynacqus ANR-24-CE40-5714-02.
AS was supported by QuantAlps and the MSCA Cofund QuanG (Grant Nr: 101081458), funded by the European Union. The views and opinions expressed are those of the authors only and do not necessarily reflect those of the European Union or the granting authority. Neither the European Union nor the granting authority can be held responsible for them. SW was supported by the DFG under EXC-2111 -- 390814868.

\bibliography{Bibliography}
		\bibliographystyle{plain}
\end{document}